\documentclass[journal]{IEEEtran}
\usepackage[dvipsnames]{xcolor}
\usepackage{bm}

\usepackage[style=plain]{floatrow}
\usepackage{url}
\usepackage{comment}

\ifCLASSINFOpdf
\else
\fi
\usepackage{graphicx}
\usepackage{cite}
\usepackage{graphicx}
\usepackage{amsmath}
\usepackage{amsfonts}
\usepackage{mathtools}

\usepackage{rotating}
\usepackage{epstopdf}
\usepackage{multirow}
\usepackage{algorithm}
\usepackage{algpseudocode}
\usepackage{cancel}
\usepackage{ulem}

\usepackage{lipsum}
\usepackage{mwe}
\usepackage{floatrow}

\newtheorem{defi}{Definition}[section]
\newtheorem{theorem}{Theorem}[section]
\newtheorem{remark}{Remark}[section]

\newtheorem{proofinner}{Proof}
\newenvironment{proof}[1][]
 {\if\relax\detokenize{#1}\relax
    \renewcommand\theproofinner{\thetheorem}%
  \else
    \renewcommand{\theproofinner}{#1}%
  \fi
  \proofinner}
 {\endproofinner}

\newcommand{\dv}[1]{{\color{black}{#1}}}

\newcommand{\nw}[1]{{\color{black}{#1}}}

\begin{document}
%

\title{Fairness and Cost Constrained Privacy-Aware Record Linkage}

%
%
%

\author{Nan Wu,
        Dinusha Vatsalan, 
        Sunny Verma,
        and~Mohamed Ali Kaafar
\thanks{Nan Wu is with the Faculty of Science and Engineering, Macquarie University, and CSIRO's Data61, Sydney, Australia. Email: nan.wu5@hdr.mq.edu.au}%
\thanks{Dinusha Vatsalan is 
with the Faculty of Science and Engineering, Macquarie University, Sydney, Australia.
Email: dinusha.vatsalan@mq.edu.au
}
\thanks{Sunny Verma is with the Faculty of Science and Engineering, Macquarie University, Sydney, Australia. E-mail: sunny.verma@mq.edu.au}
\thanks{Mohamed Ali Kaafar is with the Faculty of Science and Engineering, Macquarie University, Sydney, Australia. Email: dali.kaafar@mq.edu.au}}


%
%

\markboth{IEEE Transactions on Information Forensics and Security, February~2022}%
{Shell \MakeLowercase{\textit{et al.}}: Bare Demo of IEEEtran.cls for IEEE Journals}
%



\maketitle



%
\IEEEpeerreviewmaketitle

\begin{abstract}
Record linkage algorithms match and link records from different databases that refer to the same real-world entity based on direct and/or quasi-identifiers, such as name, address, age, and gender, available in the records. Since these identifiers generally contain personal identifiable information (PII) about the entities, record linkage algorithms need to be developed with privacy constraints. Known as privacy-preserving record linkage (PPRL), many research studies have been conducted to perform the linkage on encoded and/or encrypted identifiers. Differential privacy (DP) combined with computationally efficient encoding methods, e.g. Bloom filter encoding, has been used to develop PPRL with provable privacy guarantees. The standard DP notion does not however address other constraints, among which the most important ones are fairness-bias and cost of linkage in terms of number of record pairs to be compared. 
In this work, we propose new notions of fairness-constrained DP and fairness and cost-constrained DP for PPRL and develop a framework for PPRL with these new notions of DP combined with Bloom filter encoding. We provide theoretical proofs for the new DP notions for fairness and cost-constrained PPRL and experimentally evaluate them on two datasets containing person-specific data. Our experimental results show that with these new notions of DP, PPRL with better performance (compared to the standard DP notion for PPRL) can be achieved with regard to privacy, cost and fairness constraints.
 
\end{abstract}


\section{Introduction}
\label{sec:intro}

\IEEEPARstart{R}{ecord} 
linkage has become an essential component in any cross-organizational or cross-domain data analytics applications. 
Example applications range from healthcare applications, such as health research or personalized healthcare, business applications, such as targeted marketing and recommendation, and government services to national security applications, such as crime and fraud detection.

Due to the absence of unique entity identifiers in different databases held by different organizations (parties), linking data from different databases that correspond to the same individual needs to be conducted using the commonly available person-specific identifiers (e.g. name, address, age, and gender). However, such person-specific identifiers contain personal identifiable information (PII) about the entities, and therefore can be used to re-identify and infer information about the entities when shared across organizations. 
Linking data with privacy constraints received much attention in the literature over the last two decades. A large body of work has been done to develop privacy-preserving record linkage (PPRL) techniques, using a variety of privacy-enhancing or privacy-preserving technologies, such as cryptographic techniques and/or probabilistic techniques including \nw{Bloom Filter (BF)} encoding combined with \nw{differential} privacy (DP)~\cite{Ala12,Sch16,Xue20}.
Probabilistic encoding techniques, such as Bloom filter encoding, are computationally efficient for fuzzy linking of large-scale data, and are therefore highly suitable for Big Data applications~\cite{Vat17b}. 

Linkage is generally a classification problem that aims to classify pairs of records into the classes of `matches' and `non-matches'. Since the number of record pairs that need to be compared for the classification task becomes quadratic in the size of the databases, the records are first binned into blocks such that highly similar records are grouped together, and then records are compared with only the records in the same bin, reducing the computational complexity from quadratic to sub-quadratic. The bins of encoded records are sent to the server (third-party/linkage unit) for conducting the linkage using a classifier.

The frequency distribution of encoded records in the bins could reveal information about the bins by conducting a frequency inference attack~\cite{Vat14}. This has been addressed by several works in the literature, ranging from non-provable privacy guarantees, such as $k$-anonymous grouping~\cite{Vat13b,Vat13c,Kar12b,Ran15}, to provable privacy guarantees, such as differential privacy~\cite{Cao15,Ina08,Ina10,Kuz13,He17}. The standard DP notion for PPRL incurs computational cost in terms of additional record pair comparisons and it does not consider bias in the data, especially the fairness-bias. In this work, we consider record linkage with not only privacy constraints, but also with cost and 
fairness constraints
for practical PPRL applications.

\begin{defi}[Fairness constraint]
Fairness of linkage with regard to a certain protected sensitive feature that has $G$ sensitive groups (for example gender with $G=2$ groups, $g_1=$`male' and $g_2=$`female') determines how much the linkage classifier $f(\cdot)$ distorts from producing linkage decisions with equal probability for individuals across different protected groups, for example equal true match rates (true positive rates/TPRs) for female and male groups, i.e. $TPR_{g_1} = TPR_{g_2}$.
\end{defi}

\begin{defi}[Cost constraint]
\nw{Assuming encoded records from two databases $\mathcal{D_A}$ and $\mathcal{D_B}$ are grouped into bins using a blocking protocol $\mathcal{B}$ and blocking strategy $\mathcal{B}^\mathcal{S}$. The set of records in $\mathcal{D_A}$ been blocked into $i^{th}$ bin is represented as $\mathcal{B}_i(D_A)$.} \nw{Blocking strategy $\mathcal{B}^\mathcal{S}$ specifies the records in $\mathcal{B}_i(D_A)$ are compared with the records in $\mathcal{B}_j(D_B)$ for each pair $(i,j)\in \mathcal{B}^\mathcal{S}$.}
The computational cost of linkage is $\sum_{(i,j)\in \mathcal{B}^\mathcal{S}}\mathbb{E}(|\mathcal{B}_i(D_A)||\mathcal{B}_j(D_B)|)$.
With perturbed bins that include $r$ dummy or noisy records (where $r$ is calculated depending on the privacy budget $\epsilon$ for DP) in addition to the original encoded records from $\mathcal{D_A}$ or $\mathcal{D_B}$, the computational cost increases to
$\sum_{(i,j)\in \mathcal{B}^\mathcal{S}}\mathbb{E}((|\mathcal{B}_i(D_A)|+r_i)(|\mathcal{B}_j(D_B)|+r_j))$.
\end{defi}

Developing classifiers that are fair with respect to a protected/sensitive feature~\cite{Zaf15}, such as gender or race, is an important problem for machine learning applications in general and specifically for PPRL. \nw{This is to avoid} significant bias \nw{been introduced} towards certain groups of individuals, for example against black people in fraud and crime detection systems~\cite{Flo16,Lar16}
or online recommendation systems~\cite{Swe13}, and against women in job recommendation systems~\cite{Dat15}.

Fairness-bias in data imposes different levels of challenges on classifying record pairs into matches and non-matches for record pairs belonging to different groups. For example, let's assume that several identifiers (e.g. last name and address) exhibit more variance in female records in different databases than in their male counterparts due to marriage and/or separation, which causes record pairs belonging to the female cohort to be more difficult to classify as `matches' than male cohort. In addition, supervised machine learning classifiers can learn to ignore poor performance on a small (minority) group if it can exploit knowledge about the majority population, potentially leading to unfair outcomes.
Without careful treatment a classifier may inadvertently be biased towards the cohort that is easier to classify.

Achieving fair linkage across different groups is a difficult yet important challenge. 
Addressing privacy constraints in addition to fairness constraints for fair PPRL introduces additional challenges in terms of balancing all three key factors, which are privacy, fairness, and cost. \nw{Notably, privacy, fairness, and cost are not independent from each other. Existing works have discussed the trade-offs between privacy and communication and computational cost~\cite{He17}, and between fairness and privacy~\cite{Vat13c}. With fairness constraints considered, the trade-off between privacy, fairness, and cost needs to be addressed and balanced.} 

In this paper, we study how to address fairness and cost constraints in PPRL using fairness and cost-constrained differential privacy (DP) algorithms. We first define two new notions of DP constrained on fairness only and constrained on cost and fairness for PPRL. We propose a PPRL framework that satisfies fairness and cost constrained DP. We provide formal proofs for the two new notions of DP for PPRL and empirically study and analyse the different constraints for the PPRL problem using the proposed framework. We conducted experiments on two person-specific datasets that validate the fairness-bias in the original PPRL algorithm with the standard DP notion and the improvement in the fairness and computational cost using our proposed framework with the two new notions of DP. To the best of our knowledge, this is the first work that addresses fairness and cost constraints in DP for PPRL.  \\


\textbf{Contributions:}
\begin{enumerate}
    \item We first demonstrate that current (standard) DP notion for PPRL is unfair and biased towards minority groups of individuals
    \item 
    We then propose two new notions of DP with constraints of privacy, fairness and cost. Specifically, we formalize the two notions of fairness-constrained DP and cost-constrained fairness-aware DP for PPRL problem. 
    \item We introduce a framework enabling PPRL with fairness and cost constrained DP. Specifically, we introduce two methods that add noise to the blocks of (Bloom filter) encoded records 1) adhering to the fairness constrained DP and 2) fairness and cost constrained DP.
    \item We provide formal proofs for the two new notions of DP guarantees for PPRL and show that the two PPRL mechanisms that follow these two notions of DP provide $\epsilon$-DP guarantees constrained to cost and fairness of linkage.    
    \item \dv{We conduct experiments on two datasets, real and synthetic North Carolina Voter Registration datasets and synthetic Australian Bureau of Statistics datasets}, and evaluate the record linkage performance in terms of linkage accuracy, fairness metrics, computational cost and privacy budget and show that our proposed methods outperform the existing and baseline methods in terms of fairness and cost. 
    
\end{enumerate}

\textbf{Outline:} The rest of the paper is organized as follows: We review related work in PPRL and fairness in the following section. We then provide preliminaries of the PPRL problem and demonstrate the fairness issues in PPRL through experimental results to motivate the problem in Section~\ref{sec:prel}. Next, we formalize feature-level DP for PPRL in Section~\ref{sec:DP_PPRL} and new notions of fairness and cost constrained DP for PPRL in Section~\ref{sec:new_dp_notion}. In Section~\ref{sec:framework}, we present our framework for PPRL based on Bloom filter encoding combined with \nw{differential} private blocking methods following the new DP definitions. We present and discuss the experimental results of our algorithms in Section~\ref{sec:experiments}. Finally, we provide the takeaway messages from this work and discuss some open questions and future research directions in Section~\ref{sec:conclusion}.

\section{Related work}
\label{sec:related_work}

A long line of research has been conducted in privacy-preserving record linkage (PPRL)~\cite{Vat13} and the sub-problem of PPRL, which is private blocking to reduce the computation complexity of PPRL~\cite{Vat17b}. Only limited work has provided formal differential privacy (DP) guarantees~\cite{Kuz13,Ina08,Ina10,He17}.

For Bloom filter-based encoding used in PPRL, few studies have been conducted to provide DP guarantees to the Bloom filter encodings.
Blip is a method that flips each bit in the Bloom filter with the probability of $p = \frac{1}{1+e^{\epsilon/k}}$ to achieve $\epsilon$-differential privacy~\cite{Ala12}, 
where $k$ is the maximum number of tokens hash-mapped into the Bloom filters. 

Another method uses a flipping probability $p$ to flip the bits in the Bloom filter to meet $\epsilon$-differential privacy guarantees~\cite{Sch16}. In contrast to Blip, this method uses a parameter $p$ to control the bits to be flipped (i.e. privacy-utility trade-off) depending on the privacy budget $\epsilon$. The perturbed bit value $b_i'$ of the bit $b_i$ in a Bloom filter $b$ is (\nw{with $1 \le i \le n_l$}) $b_i' = 1$ with $p/2$ probability, $b_i'=0$ with $p/2$ probability, and $b_i'=b_i$ with $1-p$ probability. This gives $\epsilon = 2 \times k \times ln(2) \times \frac{1 - p/2}{p/2}$, \nw{where $k$ is the number of hash functions in Bloom filter}.

A recent work proposed to use $p$-value to generate a thresholded Laplace distribution in order to calculate the number of 1s in the noise vector (i.e. number of flips as 1-bits in the noise vector denote that the corresponding bits in the Bloom filter need to be flipped in the Bloom filter)~\cite{Xue20}.

A common solution to address the computation complexity of PPRL (and record linkage in general) is blocking where the records are pre-assigned into similarity groups/bins and then the comparisons are limited to only those records that are within the same bins. Bins of records can be susceptible to frequency inference attacks where the frequency distribution of blocks are compared with a known frequency distribution of external values (for example, zipcodes, if blocking is performed using zipcodes). Most existing private blocking techniques have addressed the privacy leakages using non-provable techniques, such as $k$-anonymity~\cite{Vat13b,Vat13c,Kar12b,Ran15}, pruning rare bins~\cite{Kuz13}, or using locality sensitive hashing~\cite{Kara14}. Only Few studies have addressed private blocking using differential privacy guarantees~\cite{Cao15,Ina08,Ina10,Kuz13,He17}.

\nw{Differential privacy is used to add noise into the blocks generated using hierarchical clustering~\cite{Kuz13}.} However, a recent study in~\cite{Cao15} shows that even with DP guarantees, these private blocking techniques can reveal some private information by learning from the final output of PPRL.
~\cite{He17} proposed end-to-end DP guarantees for PPRL by introducing output-constrained DP notion. In this work, DP noise is added to bins of encoded records such that the disclosure of true matching records is insensitive to the presence or absence of a single non-matching (noisy/dummy) record.
However, no work has so far studied DP constrained to fairness and/or cost. Moreover, fairness in record linkage is also an immature research topic with only one recent work in fairness-aware PPRL~\cite{Vat20a}.

There have been several algorithms and techniques proposed in the machine learning literature to improve fairness or mitigate bias in classification problems~\cite{Meh21}. These are broadly categorized into: pre-processing, in-processing (i.e. at training time), and post-processing.
The aim of pre-processing is to learn a new representation of data $X$ such that it removes the information correlated to the sensitive attribute and preserves the information of $X$ as much as possible~\cite{Zem2013,Fel15,Kra18}. The classifier can thus use the new data representation and produce results that preserve fairness. Any classifier can be supported and no re-training is required with this category of methods.

In-processing techniques add a constraint or a regularization term to the objective functions of classifiers~\cite{Aga18,Goe18,Hua19,Cel19}. Post-processing methods attempt to modify a learned classifier in a way that satisfies fairness constraints~\cite{Ple17,Woo17,Dwor18}. 
In this work, we use pre-processing techniques to add fairness and/or cost-constrained DP noise to input data (grouped into bins) to achieve fair and cost effective PPRL.

\begin{figure*}[!t]%
    \centering
     
    \includegraphics[width= 0.34\linewidth]{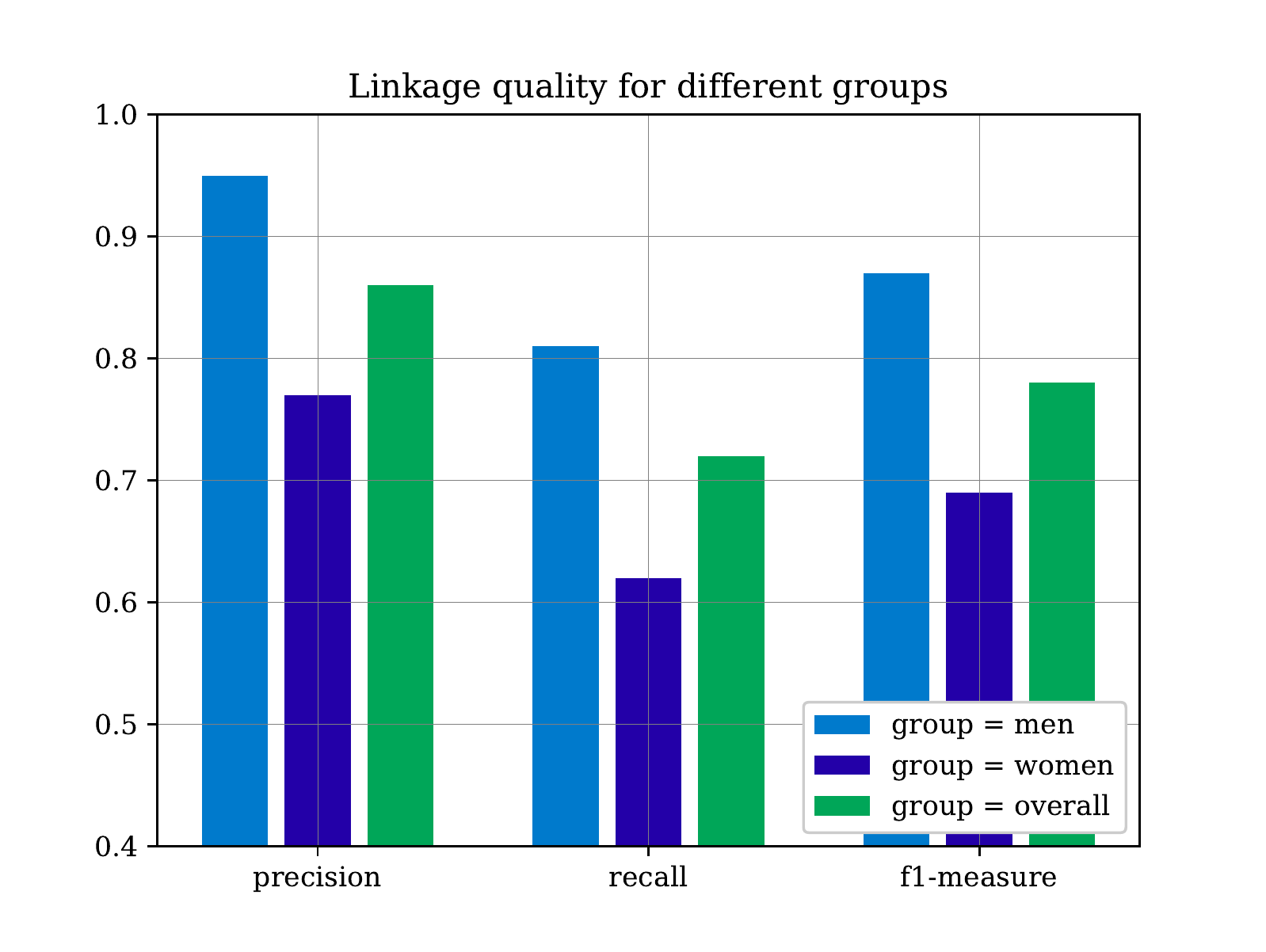}
     \includegraphics[width= 0.65\linewidth, height= 4.3cm]{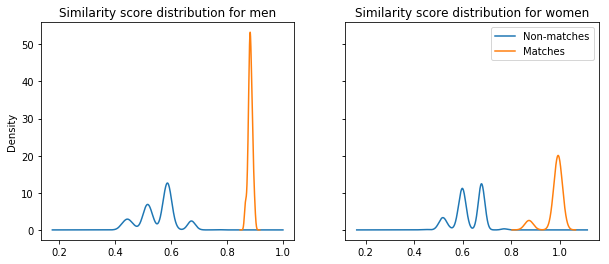}

    \caption{\dv{Similarity score distribution of true matches and true non-matches in men and women groups (right) and comparison of linkage quality (in terms of precision and recall) for men and women groups and overall on Australian Bureau of Statistics (ABS) dataset (as used in our experiments in Section~\ref{sec:experiments}) using the standard differential private blocking and logistic regression-based PPRL.} 
    }
    \label{fig:comp_acc_1}
\end{figure*}

\section{Problem Motivation}
\label{sec:prel}

We first define the PPRL problem and the general differential privacy notion for PPRL problem. We then discuss the limitations of the existing differentially private algorithms for PPRL in terms of fairness constraints using an experimental study on North Carolina Voter Registration (NCVR) dataset.

\begin{defi}[PPRL]
Assuming $p$ database owners (or parties) $P_1$, $P_2$, $\cdots$, $P_p$ with their respective databases $\mathcal{D}_1$, $\mathcal{D}_2$, $\cdots$, $\mathcal{D}_p$
(containing sensitive or confidential person-specific data), PPRL links these databases to identify \nw{whether the $x^{th}$ record in $i^{th}$ dataset $r_{i,x} \in \mathcal{D}_i$ match with the $y^{th}$ record in $j^{th}$ dataset $r_{j,y} \in \mathcal{D}_j$}, i.e. refer to the same real-world entity, where $1 \le i,j \le p$ and $i \neq j$. PPRL applies a classification function $f(\cdot)$ on the encoded records from parties that takes as input the similarity scores or distances between encoded quasi-identifying (QID) attributes of records, i.e. $sim(r_{i,x}, r_{j,y})$, where $sim(\cdot)$ is a similarity function that returns the overall similarity between two records $r_{i,x}$ and $r_{j,y}$.
\end{defi}

Without loss of generality, we assume $p=2$ in the rest of this paper and denote the two databases as $\mathcal{D_A}$ and $\mathcal{D_B}$. We assume a semi-honest (honest-but-curious) Linkage Unit ($LU$) is available to conduct the linkage on encoded records sent by the parties, which is a commonly used linkage model in many real applications~\cite{Ran13}. 
We also assume a set of QID attributes $A$ (e.g. name, address, and date of birth), which will be used for the linkage, is common to all these databases. Without loss of generality, we assume $p=2$ parties or databases in this paper and name the two parties as `Alice' and `Bob'.

Fairness of the PPRL classifier measures the classification model's behavior towards different individuals grouped by a particular protected or sensitive feature~\cite{Bin18}. The protected feature could either be part of the QIDs used to link records or not. For example, let's assume ``gender" is a protected feature dividing a dataset into two groups: male and female. Fairness of a PPRL classifier on this dataset would define whether the model treats both the male and female user groups equally in terms of correct predictions of record pairs belonging to the different groups as `matches' without giving benefit to one group more than the other. 

PPRL can result in biased predictions for different groups based on the protected feature. For example, with gender as the protected feature, female record pairs might have poor accuracy of linkage compared to male record pairs due to different levels of challenges involved in the linkage. The female group of individuals might have more likelihood of changing their last name or address than the male group due to marriage and/or separation. Additionally, if the classifier is trained on a protected feature-imbalanced dataset, then the predictions could be biased towards the minority group. These challenges impose fairness-bias in PPRL classifiers.

Fig.~\ref{fig:comp_acc_1} illustrates the fairness-bias on synthetic Australian Bureau of Statistics (ABS) dataset used for linkage experiments~\footnote{Available from \url{https://github.com/cleanzr/dblink-experiments/tree/master/data}} using the standard differential privacy-based private blocking and logistic regression PPRL classification~\cite{Kuz13,He17}. As can be seen, when data is biased towards a certain group (women/female in our experiments) in terms of errors and variations in the features used for the linkage as well as small size of the group in the training data, then the standard DP notion for PPRL exhibits fairness-bias towards the minority group (women group in this example). The similarity scores of matches and non-matches for the women group are highly overlapping than the men group (right plot in Fig.~\ref{fig:comp_acc_1}), making the linkage more challenging for the women group. Hence, the linkage quality measured in terms of precision and recall 
is considerably lower for the women group than men group (left plot in Fig.~\ref{fig:comp_acc_1}). 

\dv{There are many different fairness definitions proposed in the literature. The three commonly used definitions are Demographic Parity, Equalized Opportunity, and Equalized Odds. As discussed in~\cite{Vat20a}, Equalized Odds is the best fit fairness definition for PPRL. Since Demographic Parity requires similar rates of classification of record pairs as ‘matches’ for different groups regardless of the ground truth, it can result in linkage accuracy loss. Moreover, unlike other classification tasks, PPRL is a class-imbalanced problem with significantly lower number true matches than true ‘non-matches’, which can lead to many false positives with the Demographic Parity criteria. Equalized Opportunity only considers the true positive rate, whereas Equalized Odds considers the errors (false negatives and false positives). In PPRL we are particularly concerned about linkage errors, and therefore use Equalized Odds as the fairness criteria in our study.}

Moreover, the number of similarity comparisons required for the PPRL function increases quadratically with increasing size of datasets, and therefore blocking has been used to reduce the comparison space.
Blocking aims at reducing the comparison space for linkage by eliminating the comparisons between pair of records that are highly unlikely to be matches~\cite{Chr11,Vat17b}. 
The main aim of these techniques is to group records into disjoint or overlapping bins such that only records within the same bin need to be compared with each other. 
Differential privacy algorithms for private blocking have been developed to prevent information leakage from bins by adding dummy or noisy (encoded) records into the bins at the cost more record pair comparisons~\cite{Ina08,Ina10,Kuz13,Cao15,He17}. \dv{However, these  methods do not consider fairness-bias in the data, and thus dummy records could amplify the bias towards minority group.}

\section{Feature-level Differential Privacy for PPRL}
\label{sec:DP_PPRL}

Frequency inference attacks on encoded bins (generated by a private blocking method) can reveal information about the records in bins, for example bins with fewer records (such as bins with rare/uncommon last names, if the blocking strategy is to bin records based on last name) can be re-identified. Differential privacy noise addition has been used in the literature to make the blocks resilient against frequency attacks~\cite{He17,Ina08,Ina10,Kuz13}. These existing methods add noise, that is constrained on the standard DP guarantees, to the bins of records as dummy records.

\begin{defi}[Differentially private blocking for PPRL:]
Alice and Bob agree on a blocking function $\mathcal{B}$ with $k$ bins and strategy $\mathcal{B}^S$. A specific number of dummy records are inserted into each bin of the blocking strategy such that the bin sizes are differentially private. Each dummy record does not match \textbf{any} record. 
\end{defi}

\cite{He17} defines PPRL neighbours for record-level data. In this paper, they propose a weaker ($\epsilon, \delta$)-DP, but 
an end-to-end privacy definition for the two party setting. In their work, it is able to reveal records that are classified as `matches' and to reveal statistics about non-matching records while not revealing the presence or absence of individual non-matching records in the dataset. However, similar to other works~\cite{Ina08,Ina10,Kuz13}, their DP definition is only constrained on privacy guarantees, and does not take into account fairness and computational cost constraints.

In order to add DP noise that is constrained on fairness, we split the data into the $G$ groups of the protected feature value. We define protected feature-level DP for PPRL, where the bins of records within each protected feature group are guaranteed to be differentially private. We note that the privacy guarantees are provided for the entire record within each protected group, not only for the protected feature.
For example, if the protected feature is `gender' and it has only $G=2$ groups, which are `male' and `female', then the data is split into two disjoint groups.
So, the adjacent datasets become two male or female groups that are different by one male or female record in each group, respectively.
\dv{Please note that our proposed new DP notions are applicable to multiple protected features as well. With multiple protected features, for example gender with `female' and `male' groups and age group with `young' and `old' groups, the number of protected feature groups become $G=4$ (i.e. `young female', `old female', `young male', and `old male'. Without loss of generality, we assume a single protected feature in defining our new DP notions. }
The neighbours in each of the disjoint groups is defined based on feature-level PPRL neighbours as:

\begin{theorem}[Feature-level-PPRL neighbors]
Given function $f:\mathcal{D_A}\times\mathcal{D_B}\to \mathcal{O}$ and $D_A\in\mathcal{D_A}$, for any $1\leq g\leq G$, and any $D_B, D_{B'}\in\mathcal{D_B}$ differ in one pair of non-matching records from protected feature group $g$, $D_{B'_g}=D_{B_g}-b_g+b'_g$, $b_g\neq b'_g\neq 0$. $D_{B_g}$ and $D_{B_g'}$ are neighbors w.r.t to \nw{$f(D_{A_g},\cdot)$} for protected feature group $g$, denoted by $\mathcal{N}(f(D_{A_g},\cdot))$ if
\begin{itemize}
    \item $f(D_{A_g},D_{B_g})=f(D_{A_g},D_{B_g}')$,
    \item $D_{B_g}\setminus D_{B_g}'\cup D_{B_g}'\setminus D_{B_g}\neq \emptyset$.
\end{itemize}

\end{theorem}

\begin{proof}
If ${D_B}$ and ${D_{B'}}$ differs in a matching record, then their matching outputs with a given ${D_A}$ are different. Hence, ${D_B}$ and ${D_{B'}}$ differ in one or more non-matching records. 
Also, for any $1\leq g\leq G$, to ensure $\|D_{B_g}\|=\|D_{B'_g}\|$, the number of non-matching records added to $D_{B_g}$ to get $D_{B'_g}$ is the same as the number of non-matching records deleted from $D_{B_g}$.
Then, a neighbouring pair of $D_{B_g}$ and $D_{B'_g}$ regarding to one protected feature $g$ is differed by only one pair of non-matching pair $(b_g, b'_g)$. 

\end{proof}


\begin{defi}[Feature-level DP]
A 2-party PPRL protocol for computing function $f:\mathcal{D_A}\times\mathcal{D_B}\to \mathcal{O}$ is feature-level $(\epsilon_g,\sigma_g,f)$-differential privacy (DP) for any feature $g\in [1,\cdots,G]$ if for any $(D_{B_g},D_{B_g}')\in \mathcal{N}(f(D_{A_g},\cdot))$, the views of Alice during the execution satisfies

{\footnotesize\begin{align}
    &Pr[VIEW_A(m(D_{A_g},D_{B_g})=1)]\\\nonumber
    \leq &e^\epsilon_g  Pr[VIEW_A(m(D_{A_g},D_{B_g'})=1)] +\sigma_g, \forall g\in [1,\cdots,G]
\end{align}
}
where $m: x\times y\to \{0,1\}$ is a matching rule.
And the same holds for the views of Bob. 
\end{defi}

\begin{theorem}
The expectation of number of dummy records added regarding to each group $g$ in each bin $b$ is:

{\footnotesize\begin{align}
    \mathbb{E}(\|R_{g,b}\|) = \frac{\sigma}{2}
\end{align}}
\end{theorem}

\begin{proof}
The pdf of Laplace DP noise is     $\forall_{x\in\mathbb{R}}f(x)=\frac{1}{2\sigma} e^{-\frac{|x|}{\sigma}}$, $\sigma=\frac{\Delta B}{\epsilon}$, where $\Delta B$ is the sensitivity of the counting query to PPRL. 

{\footnotesize\begin{align*}
    \mathbb{E}(f(x)\cdot x) &= \int_{-\infty}^{+\infty} f(x)x dx\\
    &=\int_{0}^{\infty} \frac{1}{2\sigma} e^{-\frac{x}{\sigma}} x dx\\
    &=\frac{1}{2\sigma}\int_{0}^{\infty} xe^{-\frac{x}{\sigma}}dx
\end{align*}}
Let $u=-\frac{x}{\sigma}$, $\frac{du}{dx}=-\frac{1}{\sigma}$. Then, assign $dx=-\sigma du$, $x=-u\sigma$ to the equation. 
{\footnotesize\begin{align*}
\mathbb{E}(f(x)\cdot x) &= \frac{1}{2\sigma}\int \sigma^2 u e^u du\\
&=\frac{\sigma}{2}\int u e^u du\\
&=\frac{\sigma}{2}(ue^u-e^u)|_u\\
&=\frac{\sigma}{2}(-\frac{x}{\sigma}e^{-\frac{x}{\sigma}}-e^{-\frac{x}{\sigma}})|_x\\
&=-\frac{x+\sigma}{2}e^{-\frac{x}{\sigma}}|_0^\infty\\
&=\frac{\sigma}{2}
\end{align*}}

\end{proof}
\begin{theorem}
For a feature-level ($\epsilon,g$)-Differential Privacy, the overall privacy budget is $\epsilon_\ell=(\sum_{g\in [1,G]}\frac{1}{\epsilon_g})^{-1}$. 
\end{theorem}
\begin{proof}
The expectation of dummy records number for one group $g$ is $\mathbb{E}(\|R_{g,b}\|)=\frac{\Delta B}{2\epsilon_g}$. 
Because the Laplace noise for each group is independent from each other, the expectation of dummy records number for each bin $b$ is:

{\footnotesize\begin{align*}
    \mathbb{E}(\|R_{b}\|)&=\sum_{1\leq g\leq G}\mathbb{E}(\|R_{g,b}\|)\\
    &=\sum_{1\leq g\leq G} \frac{\Delta B}{2\epsilon_g}\\
\end{align*}}

Then, the expectation of dummy records number for one bin $b$ regarding to the overall privacy budget $\epsilon_\ell$ is $\mathbb{E}(\|R_{b}\|)=\frac{\Delta B}{2\epsilon_\ell}$. 
{\footnotesize\begin{align*}
    \frac{\Delta B}{2\epsilon_\ell}&=\sum_{1\leq g\leq G} \frac{\Delta B}{2\epsilon_g}\\
    &=\frac{\Delta B}{2}\sum_{1\leq g\leq G}\frac{1}{\epsilon_g}
\end{align*}}

{\footnotesize\begin{align*}
    \frac{1}{\epsilon_\ell}=\sum_{1\leq g\leq G}\frac{1}{\epsilon_g}
\end{align*}}

This concludes the proof. 
\end{proof}

\section{Fairness-aware feature-level DP for PPRL}
\label{sec:new_dp_notion}
Assume a protected feature $F$ divides the dataset into $G$ disjoint groups.
Then the fairness aware feature-level DP is defined as:

\begin{defi}[Fairness-aware feature-level DP]
A randomized mechanism $M:\Sigma\to \mathcal{O}$ satisfies ($\epsilon_{1}, \cdots, \epsilon_{G}$)- fairness aware DP if
\begin{enumerate}

    \item this mechanism satisfies $(\epsilon_g,\sigma_g,f)$-fairness aware DP,
    \item The output is constrained on 
    \begin{itemize}
    \item Fairness: The performance of the outputs between different groups of protected feature should be equalized. And Equalized Odds fairness criteria is used. 
    \item Cost: The efficiency is analysed in terms of the additional communication and computational costs introduced by DP noise for PPRL, which is the number of fake candidate pairs in this case. 
    \end{itemize}
\end{enumerate}

\end{defi}





In the following, we define the feature-level DP constrained on fairness and cost.

\subsection{Fairness constrained feature-level DP}
\label{subsec:fairness_DP}



A true match can only be an original record pair. A true non-match can be either an original record pair or a dummy record pair. The number of the dummy records is proportional to the number of the true non-matches.  
A false non-match is caused by the bias in original records in the dataset. 
The number of these false non-matches can be estimated by sampling the dataset for an approximation value. 

An original record pair can be classified as a false match if the similarity between their Bloom filters is high. The number of these false positive matches can be obtained by taking a sampling of the dataset for an estimation. On the other hand, the added dummy records can cause false positive matches. As the dummy record is created from flipping bits in original binary records, with a flipping probability \nw{$flip_g$} for feature $g\in [g_1,\cdots,g_G]$. If \nw{$flip_g$} is too small, then the similarity between the dummy record and the original record can be large enough to be classified as a false match.
For brevity, threshold based classifier is used to determine matches and non-matches, and dice coefficient is used as similarity function between two Bloom filters. 

The Dice-coefficient of two BFs (b1, b2) is calculated as: 

{\footnotesize\begin{align}
    Dice\_sim = \frac{2c}{\sum^2_{i=1}x_i} 
\end{align}}

where c is the number of common bit positions that are set to 1 in both BFs, and $x_i$ is the number of bit positions set to 1 in $b_i$, $1\le i\le 2$.

So, for threshold based classifier, when $T$ is the threshold, a record pair for feature $g$ is classified as a false match if:

{\footnotesize\begin{align}
    P(\hat{Y} = 1| g, Y = 0) = P(Dice\_sim_g > T)
\end{align}}

where the dice-coefficient $Dice\_sim_{g,dum}$ between a dummy record and its progenitor record is calculated as:

\nw{
{\footnotesize\begin{align}
    Dice\_sim_{g,dum} &= \frac{2c}{\sum^2_{i=1}x_i} \\\nonumber
    &=\frac{2(1-flip_g)n_1}{n_1+(1-flip_g)n_1+flip_g(n_{l}-n_1)} \\\nonumber
    &=\left( \frac{flip_g n_{l}}{2(1-flip_g)n_1}+1 \right) ^{-1}
\end{align}}
}

\noindent
where \nw{$flip_g$} is the flipping probability for feature $g\in [g_1,\cdots,g_G]$, $n_{l}$ is the length of one Bloom filter.
As $n_1$ is the number of bit positions that are set to 1 in the original BF. 
The value of $n_1$ follows Binomial Distribution with number of $n_{l}$ and probability $p$, $n_1\approx B(n_{l},p)$. The probability of getting exactly $n_1$ successes in $n_{l}$ independent Bernoulli trials is given by the probability mass function:

{\footnotesize\begin{align*}
P(n_1)=\binom{n_1}{n_{l}}p^{n_{l}}
\end{align*}}

\noindent
for $n_1=0,1,2,\cdots,n_{l}$, $p=\frac{1}{2}$.

With the probability distribution of $n_1$, the probability that a dummy record and its progenitor record is classified as a false match FP is:

\nw{
{\footnotesize\begin{align}
    \nw{P(FP_{g,dum}) }&= P(Dice\_sim_g > T)\\\nonumber
    &=P\left( \left( \frac{flip_g n_{l}}{2(1-flip_g)n_1}+1 \right) ^{-1} > T \right) \\\nonumber
    &=P\left( n_1>\frac{n_{l}\,T flip_g}{2(1-T)(1-flip_g)} \right) \\\nonumber
    &=\sum_{\frac{n_{l} T flip_g}{2(1-T)(1-flip_g)}}^{n_{l}} \binom{n_1}{n_{l}}p^{n_{l}}\\\nonumber
\end{align}}
}

This equation can not be simplified. Instead, as the length of the Bloom filter is large enough, where in our case the length is 300, we approximate the probability of getting $n_1$ successes in $n_{l}$ using Central limit theorem.

\begin{theorem}
\label{theorem:fp}
Suppose $\{X_1,X_2,\cdots,X_n\}$ is a sequence of i.i.d. random variables with $\mathbb{E}(X)_i=\mu$ and $Var[X_i]=\sigma^2<\infty$. Sum of the variables is ${S}=\sum_{i=1}^{i=n}{X_i}$. Then, as $n$ increases to infinity, the random variables $\frac{S-n \mu}{\sqrt{n \sigma^2}}$ converge in distribution to a normal ${\mathcal {N}}(0,1)$.
The probability of a false match with dummy record and its progenitor can be derived as follows:

\nw{\footnotesize\begin{align}
\label{eq:fp}
     P(FP_{g,dum})
     &= {\frac {1}{2}}\left[1-\operatorname{erf} \left({{\frac{ \sqrt{{n_{l}}}\,T flip_g}{2\sqrt{2}\sigma(1-T)(1-flip_g)}}-\frac{\sqrt{n_{l}}\mu  }{\sigma {\sqrt {2}}}}\right)\right]
\end{align}}
\end{theorem}

\begin{remark}
When the value of flipping probability $flip_g$ for protected feature $g$ increases, the probability of a pair with dummy record being classified as a false positive with threshold classifier decreases.  
\end{remark}

\begin{remark}
When the value of flipping probability $flip_g$ for protected feature $g$ excess a certain value, $P(FP_{g,dum})=0$.
\end{remark}

The number of false positive matches from dummy records is calculated as:

\nw{\footnotesize\begin{align}
     FP_{g,dum}
     &\propto {\frac {n_{pair\_dum}}{2}}\left[1-\operatorname{erf} \left({{\frac{ \sqrt{{n_{l}}} \,T flip_g}{2\sqrt{2}\sigma(1-T)(1-flip_g)}}-\frac{\sqrt{n_{l}}\mu  }{\sigma {\sqrt {2}}}}\right)\right]
\end{align}}

where $n_{pair\_dum}$ is the number of record pairs with at least one dummy record.

We use the following fairness loss function, which is the maximum of distance between the false positive rate of two features and distance between the false negative rate of two features:

 {\footnotesize\begin{align*}
 &fairness\_loss =  \max [abs(Pr(\hat{Y} = 1| A = a_1, Y = 0) \\\nonumber 
 &- Pr(\hat{Y} = 1 | A = a_2, Y = 0)), abs(Pr(\hat{Y} = 0| A = a_1, Y = 1) \\\nonumber 
&- Pr(\hat{Y} = 0 | A = a_2, Y = 1))],
\label{eq:fairness_loss}
 \end{align*}}

\noindent 
where $Y$ and $\hat{Y}$ are the true and predicted class labels, respectively, with two values of 1 (for matches) and 0 (for non-matches).

Fairness is calculated as:

{\footnotesize\begin{align}
  fairness &=1.0 - fairness\_loss 
\end{align}}

The false positive matches rate for feature $g$ is:

{\footnotesize\begin{align}
\label{eq:fpr}
FP_g\_rate &= \frac{FP_g}{FP_g+TN_g} \\\nonumber 
&=\frac{FP_{g,dum}+FP_{g,ori}}{FP_{g,dum}+FP_{g,ori}+TN_{g,dum}+TN_{g,ori}}\\\nonumber
\end{align}}

The original records in the dataset contribute to false negative rate and \nw{false positive} rate, while dummy records have effects on false positive rate but not false negative rate. So the false positive matches and false positive matches from original record pairs are regarded as constant values. 
By fixing the number of dummy records for different features, fairness loss can be simplified as a function of flipping probabilities \nw{$\mathcal{F}(flip_1,\cdots,flip_g)$}.

\begin{defi}[($flip_g$, Fairness)- Constrained DP]
\label{defi:fairness-constrained}
Given  $flip_g \in [flip_1 \cdots flip_g] \ge 0$, a DPRL randomised mechanism $M:\Sigma\to \mathcal{O}$ satisfies ($flip_g$,Cost)-Constrained DP if 
there exists a solution of ($flip_1^* \cdots flip_g^*$), such that
\begin{itemize}
    \item M satisfies feature-level ($\epsilon,g$)-DP;
    \item Minimises constrained optimization problem:
    \begin{align*}
        \min \mathcal{F}(flip_1,\cdots,flip_g).  
    \end{align*}
\end{itemize}
\end{defi}


\subsection{Cost constrained feature-level DP}
\label{subsec:cost_fairness_DP}
A blocking protocol $\mathcal{B}$ and blocking strategy $\mathcal{B}^\mathcal{S}$ is used to block the dataset into bins.
Given $\epsilon_1 \cdots \epsilon_g$ for $G$ protected groups of the sensitive feature, the expected number of dummy records that need to be added for each of these groups for DP guarantees per bin $b$ $\mathbb{E}(\|{R}_{g,b}^+\|)$ is a constant denoted by $C_{g,b}$, where $1 \leq g \leq G$, $1 \leq b \leq k$, and ${R}_{g,b}$ corresponds to the dummy records of group $g$ from bin $b$. 

By fixing the flipping probability for each feature $g\in [1,\cdots,G]$, fairness loss function can be simplified as a function of privacy budgets $\mathcal{F}(\epsilon_1 \cdots \epsilon_g)$. 


The number of candidate matches is a random variable, denoted by $\mathbb{C}$, with expected value

{\footnotesize\begin{align}
    \mathbb{E}(\mathbb{C})= \sum_{(b_A,b_B)\in\mathcal{B^S}} \mathbb{E}(\|\Tilde{Bin}_{b_A}(D_A)\|)  \mathbb{E}(\|\Tilde{Bin}_{b_B}(D_B)\|)
\end{align}}


The expected privacy cost $\mathbb{C}_{g,dum}$ for group $g$ is:

{\footnotesize\begin{align}
    \mathbb{E}(\mathbb{C}_{g,dum}) 
    & = \sum_{(b_A,b_B)\in\mathcal{B^S}}N_{b_A,g}C_{b_B,g}
    +\sum_{(b_A,b_B)\in\mathcal{B^S}}N_{b_B,g}C_{b_A,g}\\\nonumber &+\sum_{(b_A,b_B)\in\mathcal{B^S}}C_{b_A,g}C_{b_B,g}\\ 
    & \approx\sum_{(b_A,b_B)\in\mathcal{B^S}}N_{b_A,g}\mathbb{E}(\|R_{g,b_B}\|)\\\nonumber
    &+\sum_{(b_A,b_B)\in\mathcal{B^S}}N_{b_B,g}\mathbb{E}(\|R_{g,b_A}\|)\\\nonumber
    &+\sum_{(b_A,b_B)\in\mathcal{B^S}}\mathbb{E}(\|R_{g,b_A}\|)\mathbb{E}(\|R_{g,b_B}\|)\\\nonumber
    &=(N_{A,g}+N_{B,g})\frac{\Delta B}{2\epsilon_g}
    +N_{bins}\frac{\Delta B^2}{4\epsilon_g^2}
\end{align}}

The expected False positive rate from Equation.~\ref{eq:fpr} is:
{\footnotesize\begin{align}
\label{eq:fpr2}
    FP_g\_rate = \frac{P(FP_{g,dum})\mathbb{E}(\mathbb{C}{g,dum})+FP_{g,dum}}{FP_{g,dum}+FP_{g,ori}+TN_{g,dum}+TN_{g,ori}}
\end{align}}

\hspace{1pt}

\begin{defi}[($\epsilon$,Cost)-Constrained feature-level DP]
\label{defi:cost-constrained}
Given  $\epsilon_g \in [\epsilon_1 \cdots \epsilon_G] \ge 0$, a DPRL randomised mechanism $M:\Sigma\to \mathcal{O}$ satisfies ($\epsilon_\ell$,Cost)-constrained feature-level DP if 
there exists a solution of ($\epsilon_1^* \cdots \epsilon_g^*$), such that
\begin{itemize}
    \item M satisfies ($\epsilon^*,g$)-feature-level DP;
    \item ${\epsilon_\ell}^{-1}=\sum_{1\le g\le G}{\epsilon_g}^{-1}$;
    \item fairness loss is minimised. 
\end{itemize}
\end{defi}

It is obvious that the privacy budget for each group and cost are conflicting with each other. So we use Lagrangian function to optimize this problem. Consider the following optimization problem minimizing the overall cost with privacy budget constraints:

{\footnotesize\begin{align}
    &\min \mathcal{F}(\epsilon_1 \cdots \epsilon_g)\\
    &s.t.
    \sum_{1\le g\le G}{\epsilon_g}^{-1}- {\epsilon_\ell}^{-1}=0 
\end{align}}

\begin{figure*}[!t]%
    \centering
    \includegraphics[width=0.85\linewidth,trim={2cm 0 2cm 0}]{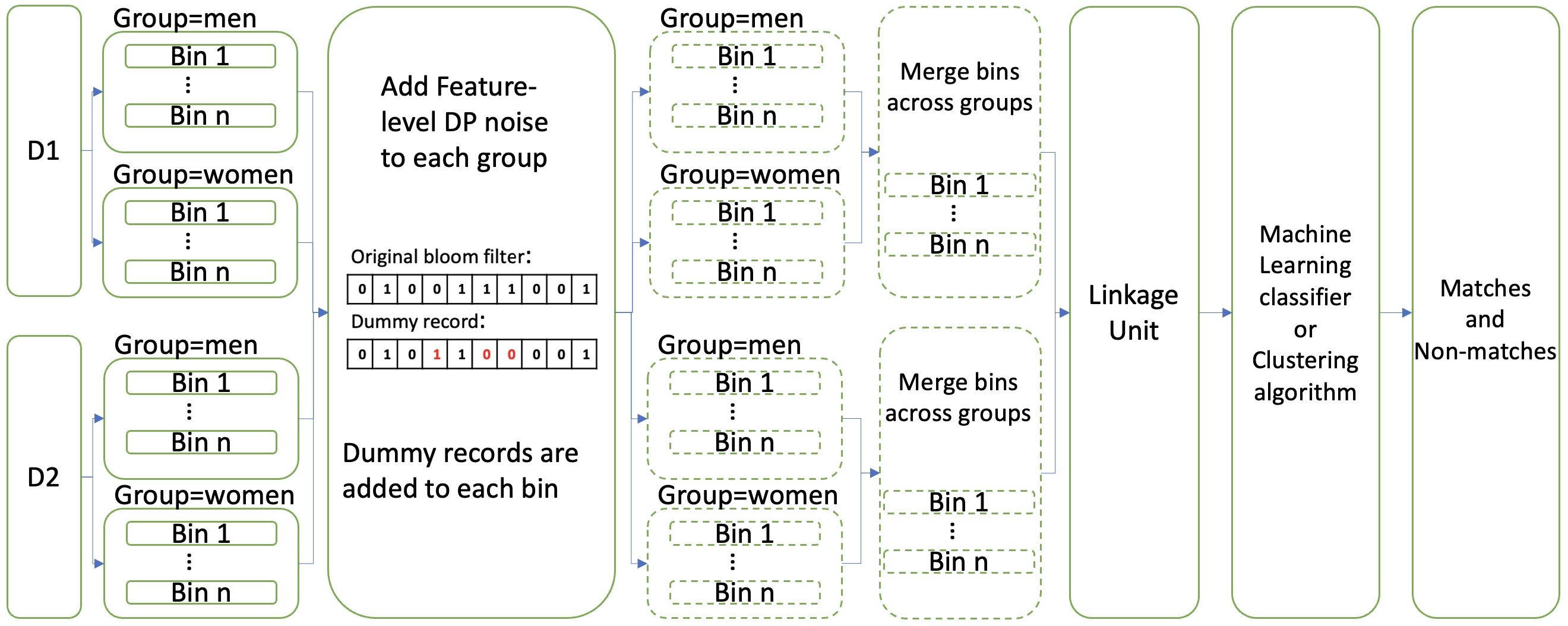}
    \caption{Proposed framework for PPRL with fairness and cost-constrained DP, with $G=2$, $g_1=$`men' and $g_2=$`women'}
    \label{fig:framework}
\end{figure*}

\section{Framework for PPRL with Fairness and cost constrained DP} 
\label{sec:framework}
 
In this section, we present a framework for PPRL based on Bloom filter encoding and the new notions of \nw{differential privacy (DP)} for fairness and computational cost-aware PPRL. Records are encoded into Bloom filters at the data owners site. The encoded records are then sent to a linkage unit which performs All-Pairwise Comparisons (APC) to determine the linked or matched pairs of records.

Blocking techniques are useful as a pre-processing step prior to APC to achieve efficiency and high recall in record linkage. 
In PPRL with DP Blocking, DP hides the presence or absence of a single record, and hence the number of candidate matches stays roughly the same on $D_b$ and $D_b'$ that differ in a single record. It provides a strong end-to-end privacy guarantee – it leaks no information other than the sizes of the databases and the set of matching records. So, the number of dummy records that are added to each blocked bin is uniformly distributed to different groups. Considering fairness-bias in data, we introduce PPRL with feature level DP Blocking method, with which the dummy records added for different groups are biased and fairness-aware:

\begin{algorithm}[t]
\footnotesize
    \caption{\textit{Feature-level DP blocking method}}
    \hspace*{\algorithmicindent} \textbf{Input:} $rec_i$, $i=1,\cdots,n$ n is the size of original dataset $X$, $bin\_label_b$, $b = 1,\cdots,B$ B is the number of bins, $nb$ is length (in bits) of bin label,  $\epsilon_i$, $Pr\_flip_g$, $g=1,...,G$ G is the total number of protected features. 
    \begin{algorithmic}[1]
    \State Select $n_b$ number of bits in Bloom filter and mark as Bloom filter labels. 
    \For{$i=1,\cdots,n$}\Comment{Do for each record in dataset $X$}
    \State allocate $X_i$ to each bin according to its Bloom filter label.
    \EndFor
    \For{$b = 1,\cdots,B$}\Comment{Do for each bin}
    \For{$g = 1,\cdots,G$}\Comment{Do for each protected group}
    \State Obtain $n\_dummy_g$ number of dummy records from $\epsilon_g$-DP noise.
    \If{$n\_dummy\_g > n\_b_g$} \Comment{$n\_dummy\_g$ is less or equal to the number of records with protected feature g in the bin}
    \State $n\_dummy\_g \gets n\_b_g$
    \EndIf
    \If{$n\_dummy\_g<0$}
    \State $n\_dummy\_g \gets 0$ \Comment{No records are deleted}
    \EndIf
    \State Randomly select $n\_dummy\_g$ original records of protected feature g from current bin. 
    \State Flip Bloom filter bits for each record with flipping probability $Pr\_flip\_g$. \Comment{Create dummy records} 
    \State Add dummy record to current bin. 
    \EndFor
    \EndFor  
    \end{algorithmic}
\end{algorithm}

\begin{itemize}
    \item The number of all dummy records added is calculated based on all groups.
    \item The number of dummy records added as one group label is calculated by feature level DP noise. 
    \item Given the dummy records is generated from modifying original record by flipping each bit in Bloom filter with a probability, The data bias of dummy records is also fairness aware regarding to the flipping probability. 
    \item Dummy records for female and male can be manipulated separately both on numbers of dummy records and error rate/flipping probability. 
    \item In each bin, the number of dummy records for each group is independent from each other. Thus, the added feature level DP noise is fairness aware and hide records from different groups. 
\end{itemize}

An overview of our proposed framework for fairness-aware PPRL with Bloom filter encoding and new notions of DP is shown in Fig.~\ref{fig:framework}. The proposed feature-level DP blocking algorithm for PPRL with new DP notions is outlined in Algorithm~1, and is described in the following sub-section.

\subsection{Algorithm Description}

In this algorithm, the two data owners (parties) Alice and Bob agree on a blocking function $\mathcal{B}$ with $k$ bins and blocking strategy $\mathcal{B}^S$. Then, a chosen number of dummy records are added to each bin to make the bin sizes deferentially private. The dummy records are carefully created such that they do not match with any record. Linkage is done using any machine learning classifier (e.g. logistic regression classifier) or simple threshold-based classifier.

\section{Experimental Evaluation}
\label{sec:experiments}

We conducted our experiments on three-pairs of datasets sampled from two sources/domains which contain person-specific data: 1) Australian Bureau of Statistics (ABS) datasets and 2) North Carolina Voter Registration (NCVR) datasets.

\begin{enumerate}

\item \textbf{ABS}: This is a synthetic dataset used internally for linkage experiments at the Australian Bureau of Statistics (ABS)~\footnote{Available from \url{https://github.com/cleanzr/dblink-experiments/tree/master/data}}. It simulates an employment census and two supplementary surveys. We sampled 5000 records for two parties with area level (categorical), mesh block (categorical), sex, industry (multi-categorical), and part time / full time (binary-categorical).

\begin{figure*}[!t]
\centering
        \CommonHeightRow{%
            \begin{floatrow}[2]%
                \ffigbox
                {\includegraphics[width=0.36\textwidth]{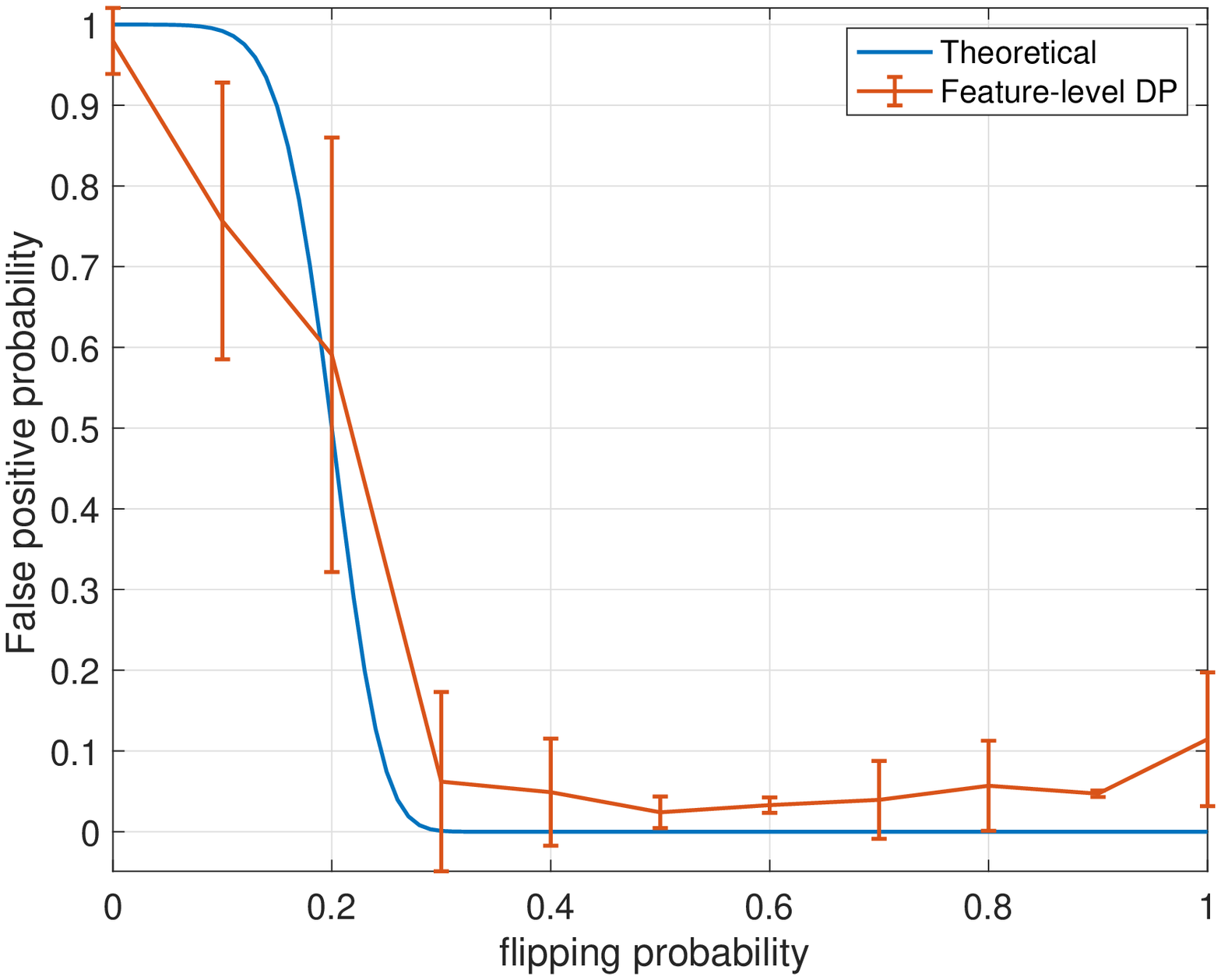}}
                {\vspace{-0.15cm}\caption{False Positive probability theoretical prediction in Theorem.~\ref{theorem:fp} and simulation results versus flipping probability with threshold classifier ($t=0.8$) on NCVR-no mod dataset}
                \label{fig:fp_probability}}
                \ffigbox
                {\includegraphics[width=0.36\textwidth]{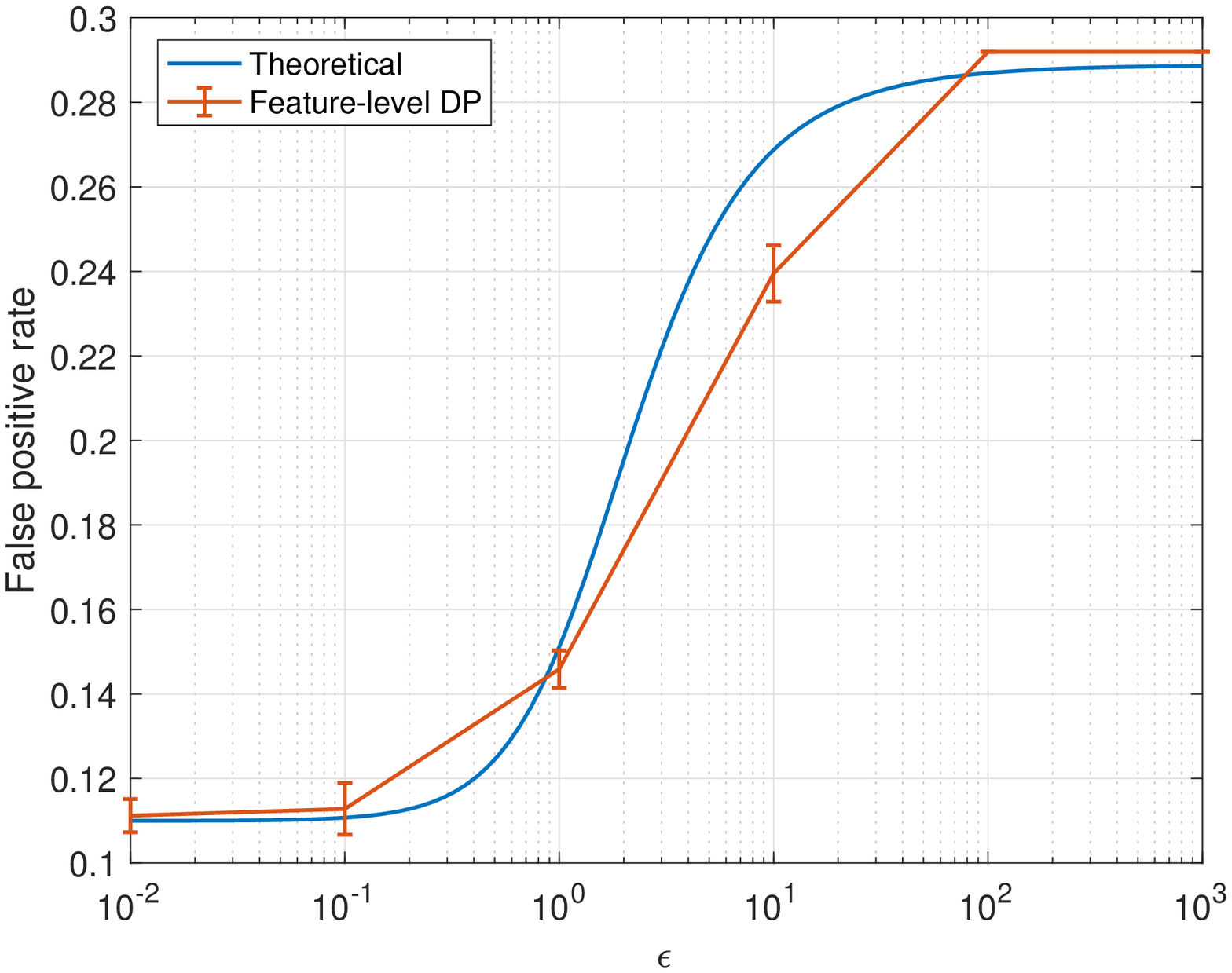}}
                {\vspace{-0.15cm}\caption{False Positive rate theoretical prediction in Equation.~\ref{eq:fpr} and simulation results versus privacy budget $\epsilon$ with threshold classifier ($t=0.8$) (left) and logistic linear regression classifier (right) on NCVR-no mod dataset}
                \label{fig:fp_rate}}
            \end{floatrow}}%
\end{figure*}

\item \textbf{NCVR-no-mod}: We extracted a pair of datasets with 5000 records each and a pair of datasets with 10000 records each for two parties from the \dv{real} North Carolina Voter Registration
(NCVR) database~\footnote{Available from
\url{http://dl.ncsbe.gov/data/}} 
with 50\% of matching
records between the two parties. 
Ground-truth
is available based on the voter registration identifiers. We used given name (string), surname (string), suburb (string), postcode (string), and gender (categorical) attributes for the linkage.

\item \textbf{NCVR-mod}: We generated another series of \dv{synthetic} NCVR datasets for each pair of NCVR datasets generated above, 
where we included 50\% synthetically
corrupted records
using the GeCo tool~\cite{Tra13}.
We applied various corruption functions from the GeCo tool on randomly selected attribute values, including character edit operations
(insertions, deletions, substitutions, and transpositions), and
optical character recognition and phonetic modifications
based on look-up tables and corruption rules~\cite{Tra13}. This allows us to evaluate how real data errors impact the linkage quality and fairness of linkage decisions.

\end{enumerate}

\dv{In our experiments, we use two binary protected features, which are gender with `male' and `female' groups, and age group with `young' ($<= 45$) and `old' ($> 45$) groups. 
}
The evaluation metrics used to measure the matching or linkage performance, fairness of linkage, computational efficiency, and privacy guarantees are:
\begin{enumerate}
    \item \textbf{Matching performance:}
\dv{Given the number of true and false positives/matches $TP$ and $FP$, and the number of true and false negatives/non-matches $TN$ and $FN$ predicted by the classifier:}
    \begin{enumerate}
        \item \textbf{False positive rate} is the proportion of true negatives/non-matches that are predicted by the linkage classifier as positives/matches, i.e. $\frac{FP}{FP+TN}$.
        \item \textbf{Precision} is the percentage of correctly classified matches against all pairs that are classified as matches: $precision = \frac{TP}{TP + FP}$. It is also known as true positive rate (TPR).
        \item \textbf{Recall} is the percentage of correctly classified matches against all true matches: $recall = \frac{TP}{TP + FN}$. \footnote{\dv{Please note that while Bloom filter encoding does not admit $FN$s (i.e. $FN=0$ and hence recall is $1.0$), the recall of the linkage algorithm may not be $1.0$ due to Differential privacy noise addition, the blocking/binning quality, or data errors and variations that could lead to recall loss.}}
        \item \dv{\textbf{$F^*$-measure~\cite{Han21}} has been used in the recent record linkage literature as an alternative to $f1$-measure, as the $f1$-measure has a limitation of appropriately measuring the linkage quality due to the relative importance given to precision and recall. $F^*$-measure is calculated as $F^* = \frac{TP}{(TP + FP + FN)}$.}
    \end{enumerate}

        \item \textbf{Fairness:}
        \begin{enumerate}
            \item \textbf{Fairness loss}: Equalized Odds fairness criteria is used and the fairness loss based on equalized odds criteria is calculated as the maximum of the absolute difference between FPR of the two groups (based on the protected feature, e.g. male and female groups) and the absolute difference between FNR of the two groups. 
            
            {\footnotesize\begin{align*}
            fairness\_loss =  &\max [
            abs(Pr(\hat{Y} = 1| A = a_1, Y = 0) \\\nonumber 
            &- Pr(\hat{Y} = 1 | A = a_2, Y = 0)),\\\nonumber 
            &abs(Pr(\hat{Y} = 0| A = a_1, Y = 1) \\\nonumber 
            &- Pr(\hat{Y} = 0 | A = a_2, Y = 1))],
           \label{eq:fairness_loss}
            \end{align*}}
            
            where $Y$ and $\hat{Y}$ are the true and predicted class labels, respectively, with two values of 1 (for matches) and 0 (for non-matches). Fairness is calculated as:
           
            {\footnotesize\begin{equation}
            fairness = 1.0 - fairness\_loss
            \label{eq:fairness}
            \end{equation}}
        \end{enumerate}
    \dv{    
    \item \textbf{Computational cost:}
    \begin{enumerate}
        \item \textbf{Number of record pair comparisons:} Computational cost of linkage is measured as the required number of record pair similarity comparisons (including real and dummy pairs).
    \end{enumerate}
    }
    \dv{    
    \item \textbf{Privacy:}
    \begin{enumerate}
        \item \textbf{Differential privacy budget $\epsilon$:} privacy is measured using the privacy budget $\epsilon$ for Differential privacy guarantees. This is one of the widely used metrics under the indistinguishability category, as described in the taxonomy proposed in~\cite{Wag18}).
    \end{enumerate}
    }    
\end{enumerate}

\begin{figure*}[!t]
    \centering
    \includegraphics[width=0.99\linewidth]{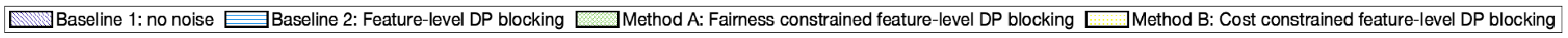} \vspace{0.5cm}\vfill
    \RawFloats
    \begin{minipage}{0.49\linewidth}
    \includegraphics[width=0.49\linewidth,trim={2.5cm 2.cm 2.5cm 2.cm},clip]{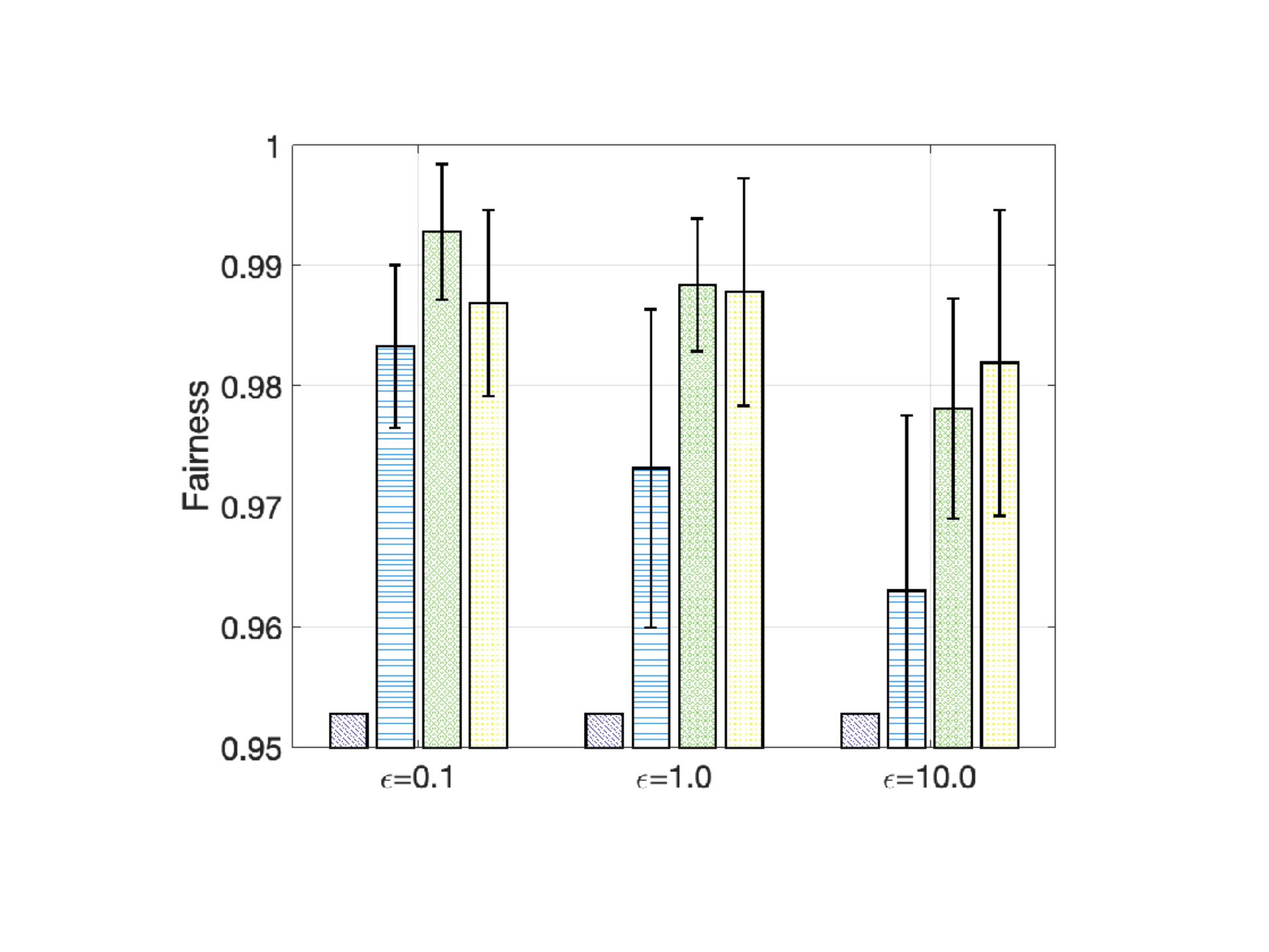}
    \includegraphics[width=0.49\linewidth,trim={2.5cm 2.cm 2.5cm 2.cm},clip]{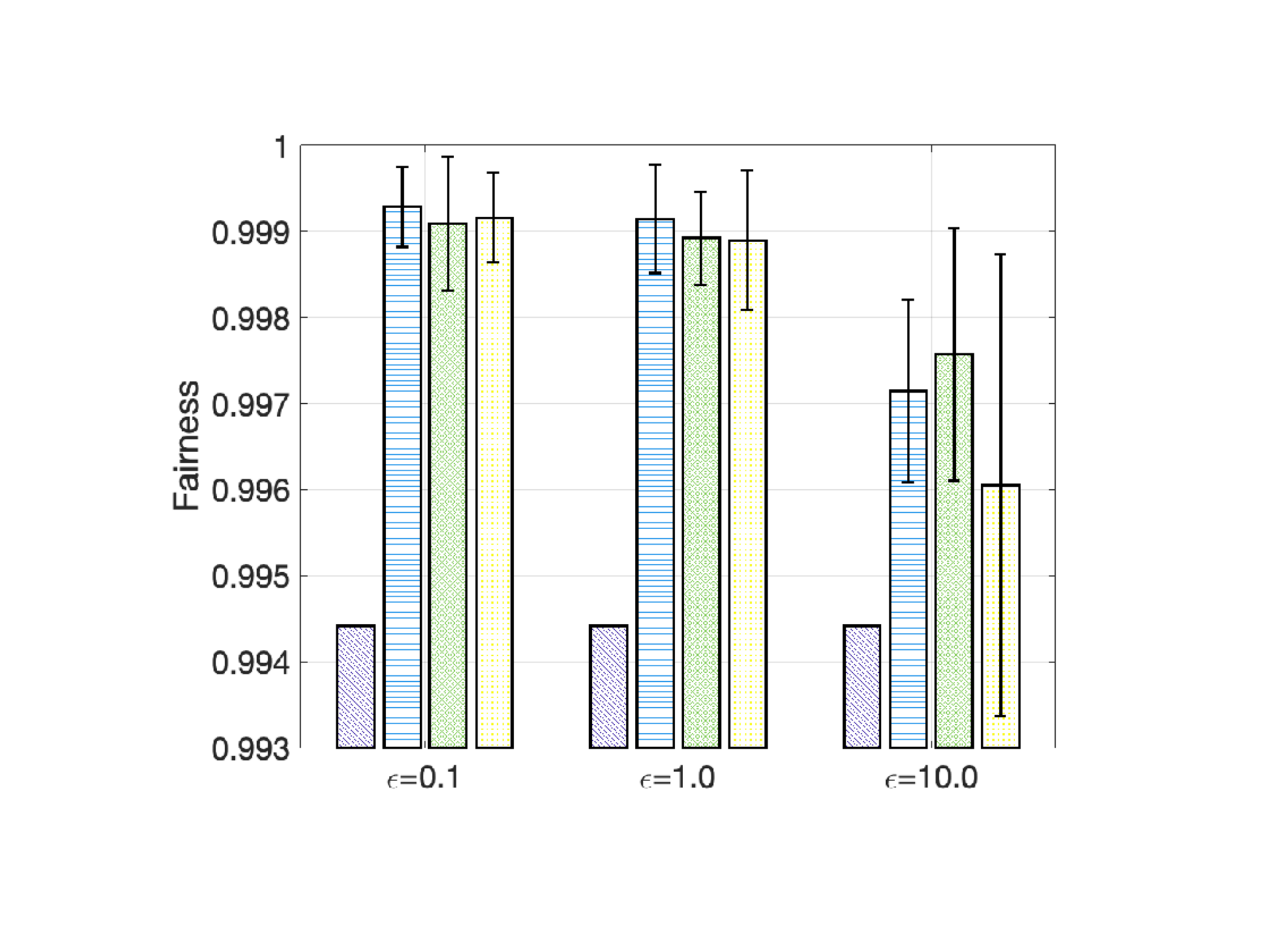}
   \caption{\nw{Fairness comparison between four scenarios and three privacy levels with threshold classifier ($t=0.8$) (left) and logistic linear regression classifier (right) on NCVR-no mod dataset with gender as sensitive feature}}
    \label{fig:nomod_fairness}
    \end{minipage}\hfill
    \begin{minipage}{0.49\linewidth}
    \includegraphics[width=0.49\linewidth,trim={2.5cm 2.cm 2.5cm 2.cm},clip]{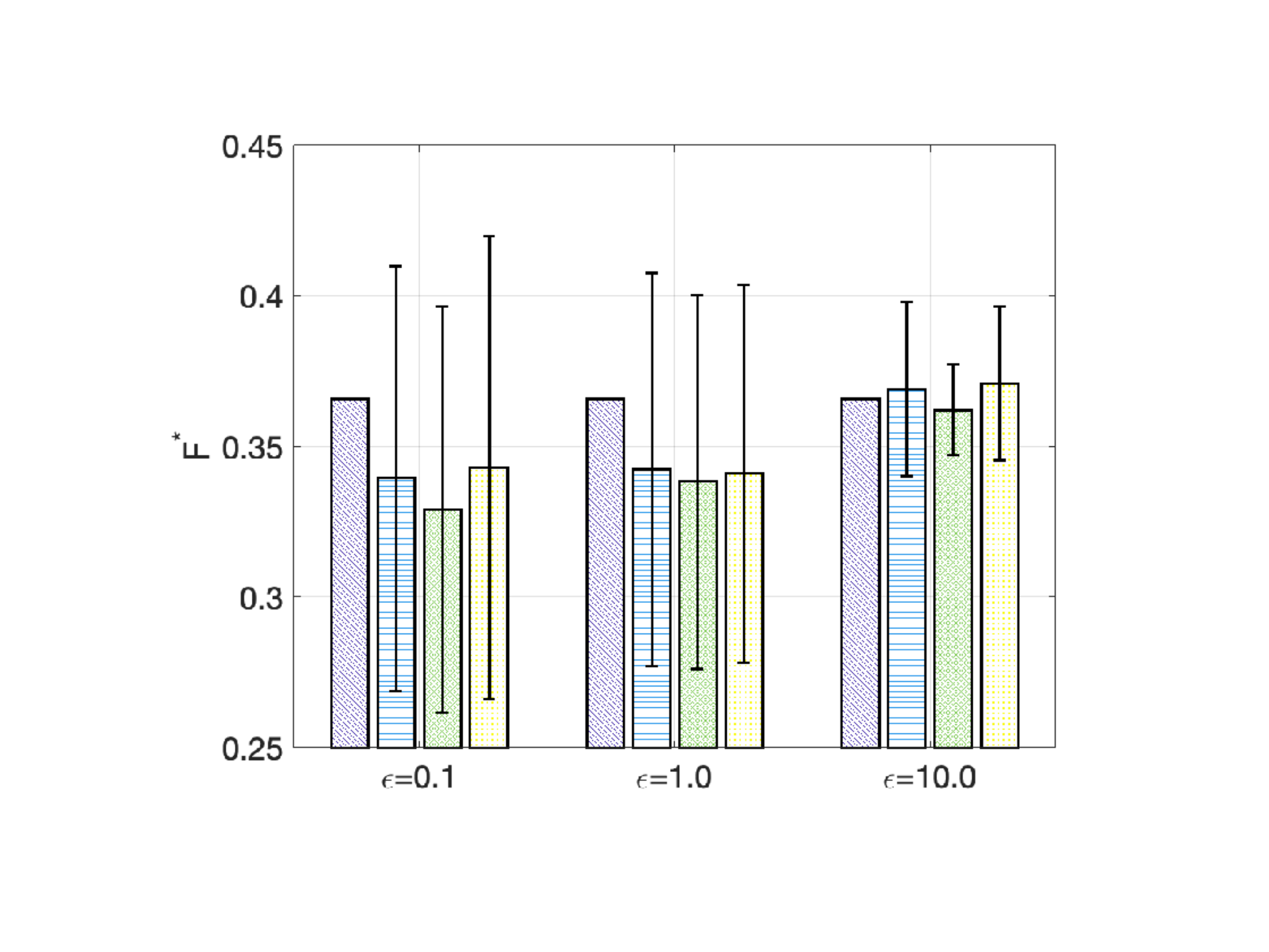}
    \includegraphics[width=0.49\linewidth,trim={2.5cm 2.cm 2.5cm 2.cm},clip]{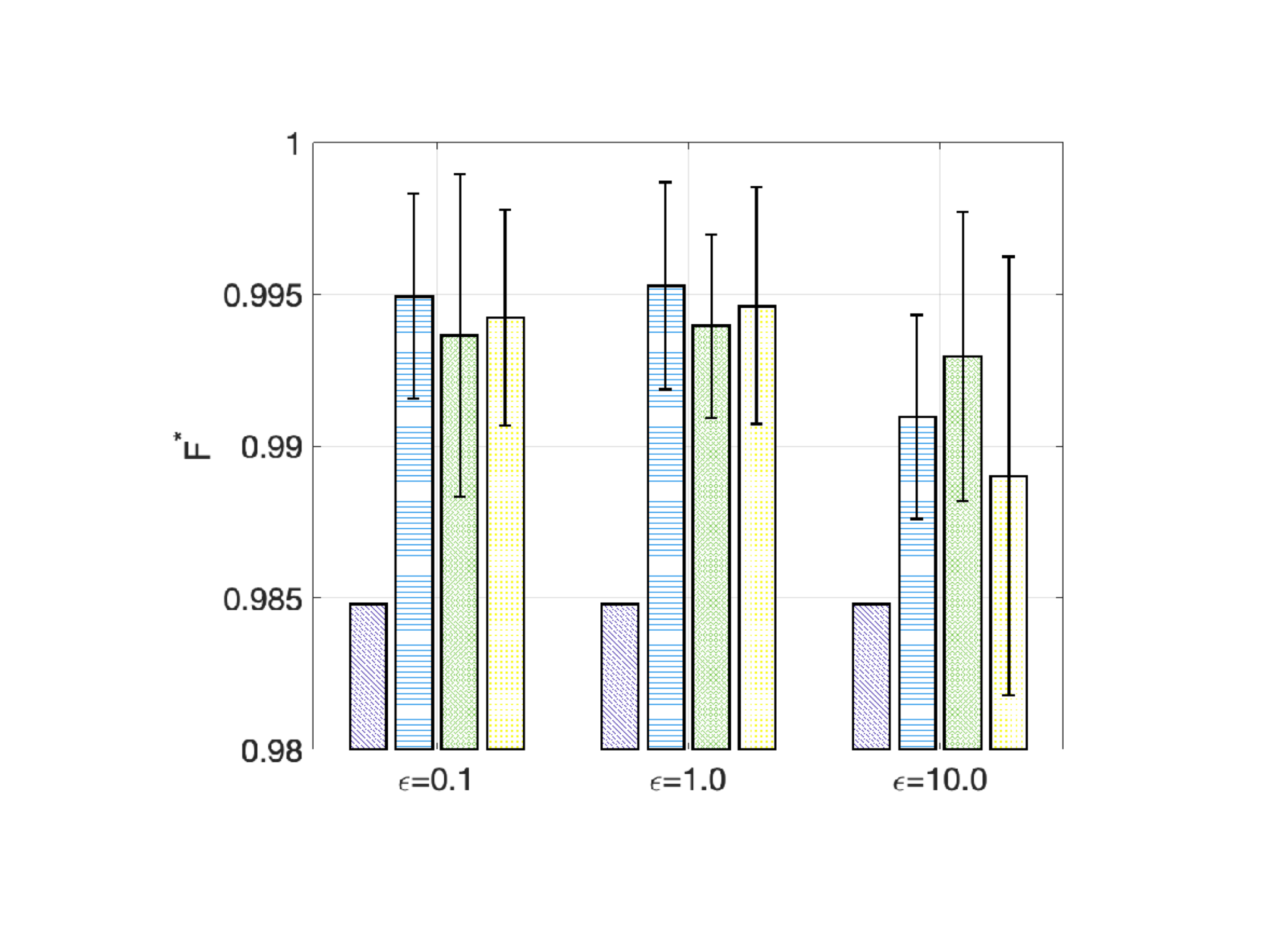}
    \caption{\nw{$F^*$-measure comparison between four scenarios and three privacy levels with threshold classifier ($t=0.8$) (left) and logistic linear regression classifier (right) on NCVR-no mod dataset with gender as sensitive feature}}
    \label{fig:nomod_f1}
    \end{minipage}
\end{figure*}

\begin{figure*}[!t]
    \centering
    \RawFloats
    \begin{minipage}{0.49\linewidth}
    \includegraphics[width=0.49\linewidth,trim={2.5cm 2.cm 2.5cm 2.cm},clip]{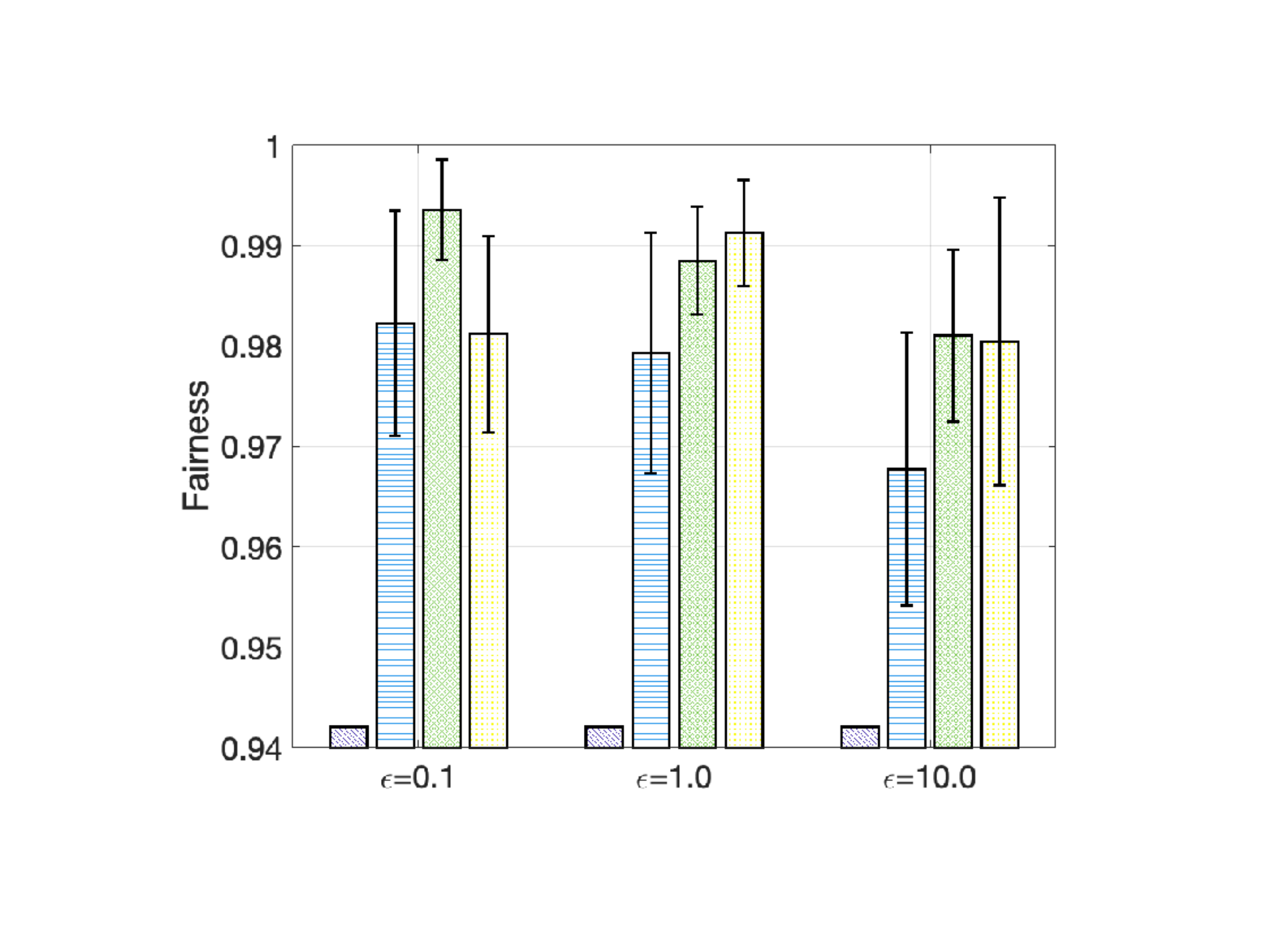}
    \includegraphics[width=0.49\linewidth,trim={2.5cm 2.cm 2.5cm 2.cm},clip]{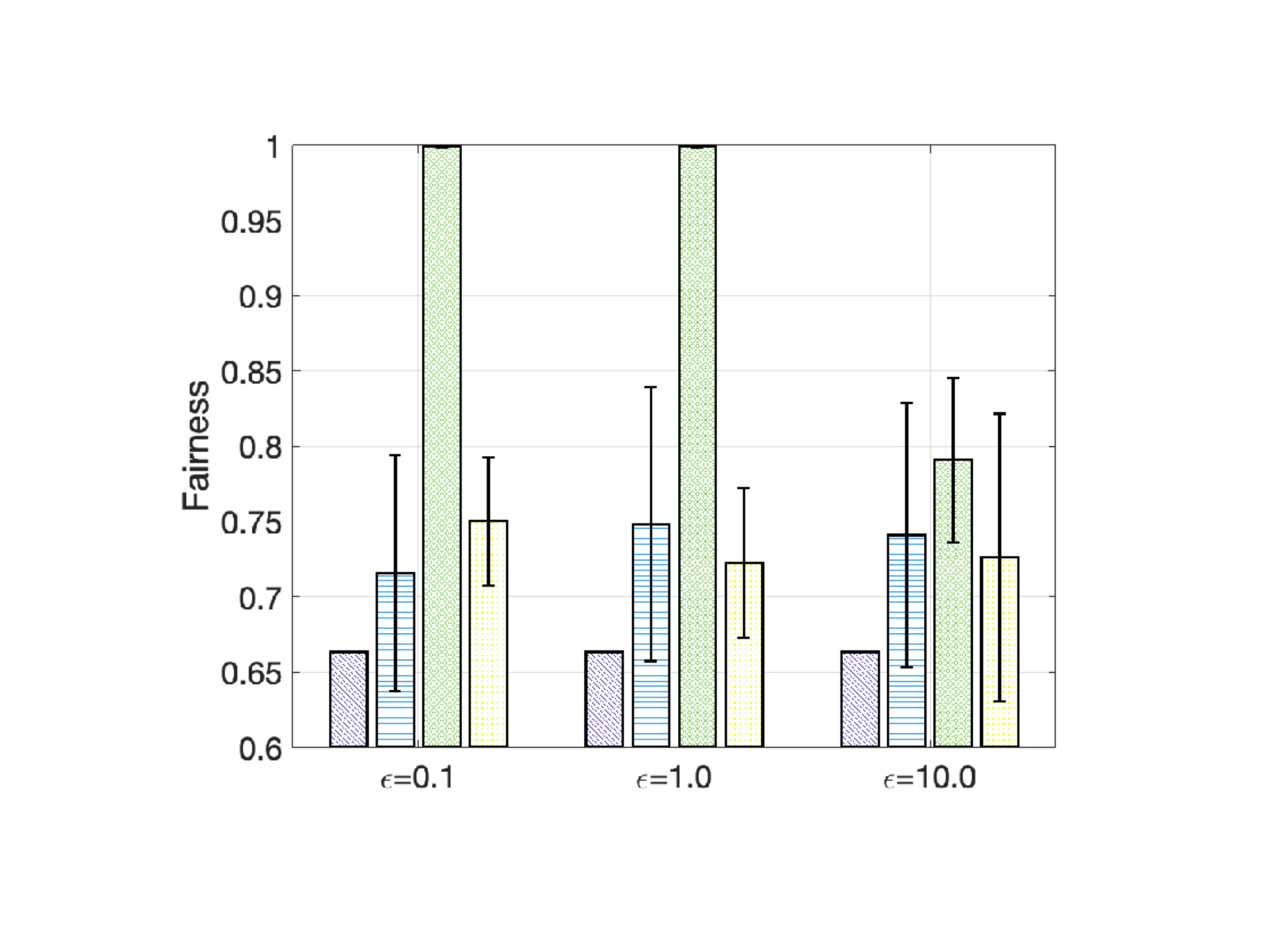}
   \caption{\nw{Fairness comparison between four scenarios and three privacy levels with threshold classifier ($t=0.8$) (left) and logistic linear regression classifier (right) on with NCVR-mod dataset with gender as sensitive feature}}
    \label{fig:withmod_fairness}
    \end{minipage}\hfill
    \begin{minipage}{0.49\linewidth}
    \includegraphics[width=0.49\linewidth,trim={2.5cm 2.cm 2.5cm 2.cm},clip]{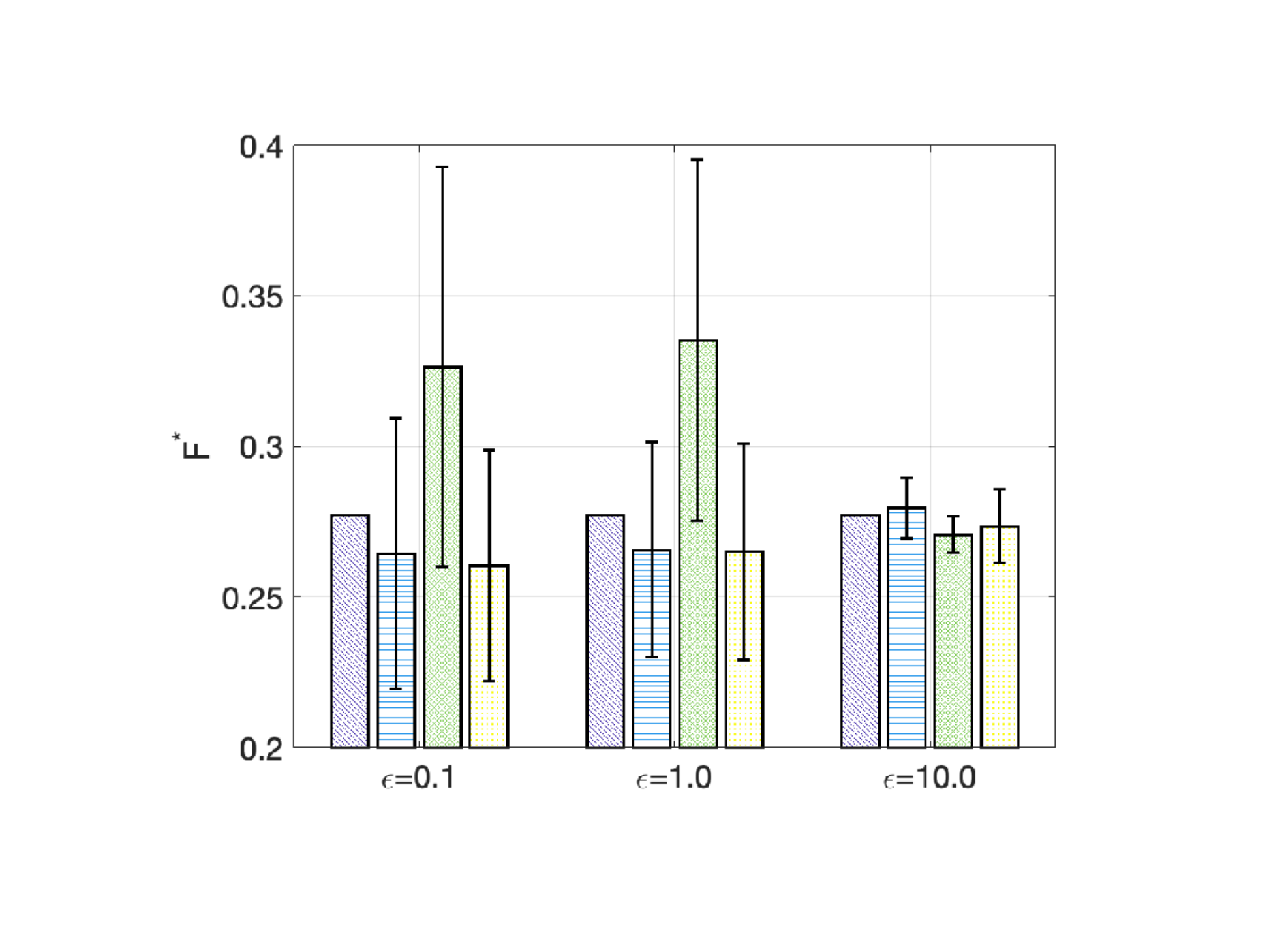}
    \includegraphics[width=0.49\linewidth,trim={2.5cm 2.cm 2.5cm 2.cm},clip]{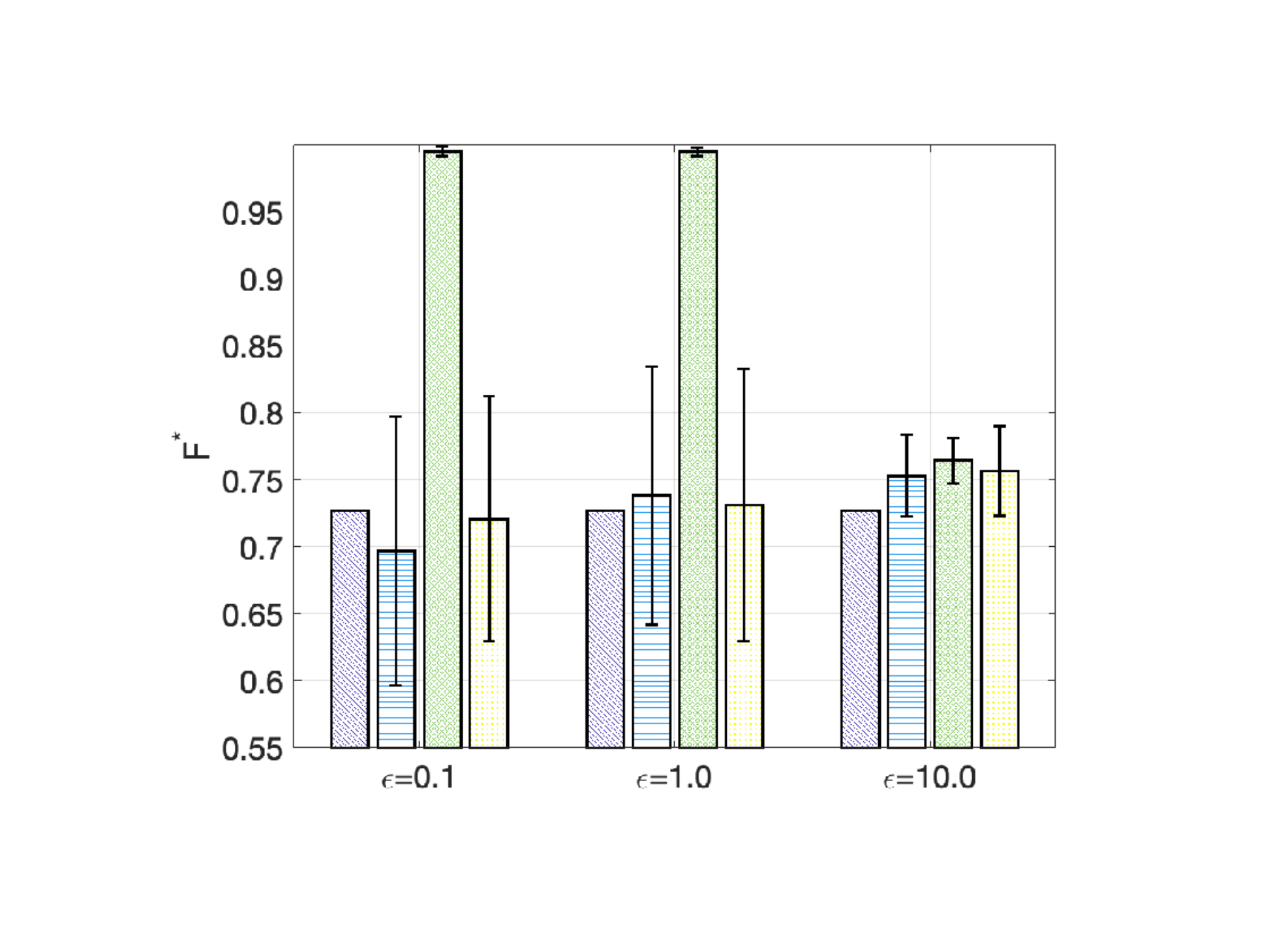}
   \caption{\nw{$F^*$-measure comparison between four scenarios and three privacy levels with threshold classifier ($t=0.8$) (left) and logistic linear regression classifier (right) on with NCVR-mod dataset with gender as sensitive feature}}
    \label{fig:withmod_f1}
    \end{minipage}
\end{figure*}

\dv{Fairness and matching performance metrics are in the range of $[0,1]$ and the higher the values for all metrics except the $fairness\_loss$ metric (for which the lower the value), the better the performance of the linkage in terms of accurate and fair linkage. Lower values for number of record comparisons and privacy budget indicate better computational efficiency and privacy guarantees.}

In this work, we used threshold classifier and Logistic regression machine learning model to perform the linkage (i.e. to classify matches and non-matches). 
We used logistic regression classifier available in sklearn library in Python for classifying record pairs.

First, we demonstrate and validate the false positive probability in Equation.~(\ref{eq:fp}) and false positive rate predictions in Equation.~(\ref{eq:fpr2}). 
Consider the case when there are 5000 entities in each dataset with gender as sensitive features, feature-level DP blocking method is used in the record linkage process with equal privacy budgets and equal flipping probability between different genders. The length of the bloom filter used is 300. The number of hash functions used in bloom filters is 30. The number of iterations of bloom filter is 5. The length of sub-strings (q-grams) is 2. The length of label for each blocked bin is 30. The number of iterations of encoding is 2. The flipping probability varies in the range $[0.0,1.0]$. Threshold classifier with threshold $t=0.8$ is used after the linkage unit to classify matches and non-matches.

\begin{figure*}[!t]
    \centering
     \includegraphics[width=0.99\linewidth]{legend_2.png}  \vspace{0.5cm}\vfill
    \RawFloats
    \begin{minipage}{0.49\linewidth}
    \includegraphics[width=0.49\linewidth,trim={2.5cm 2.cm 2.5cm 2.cm},clip]{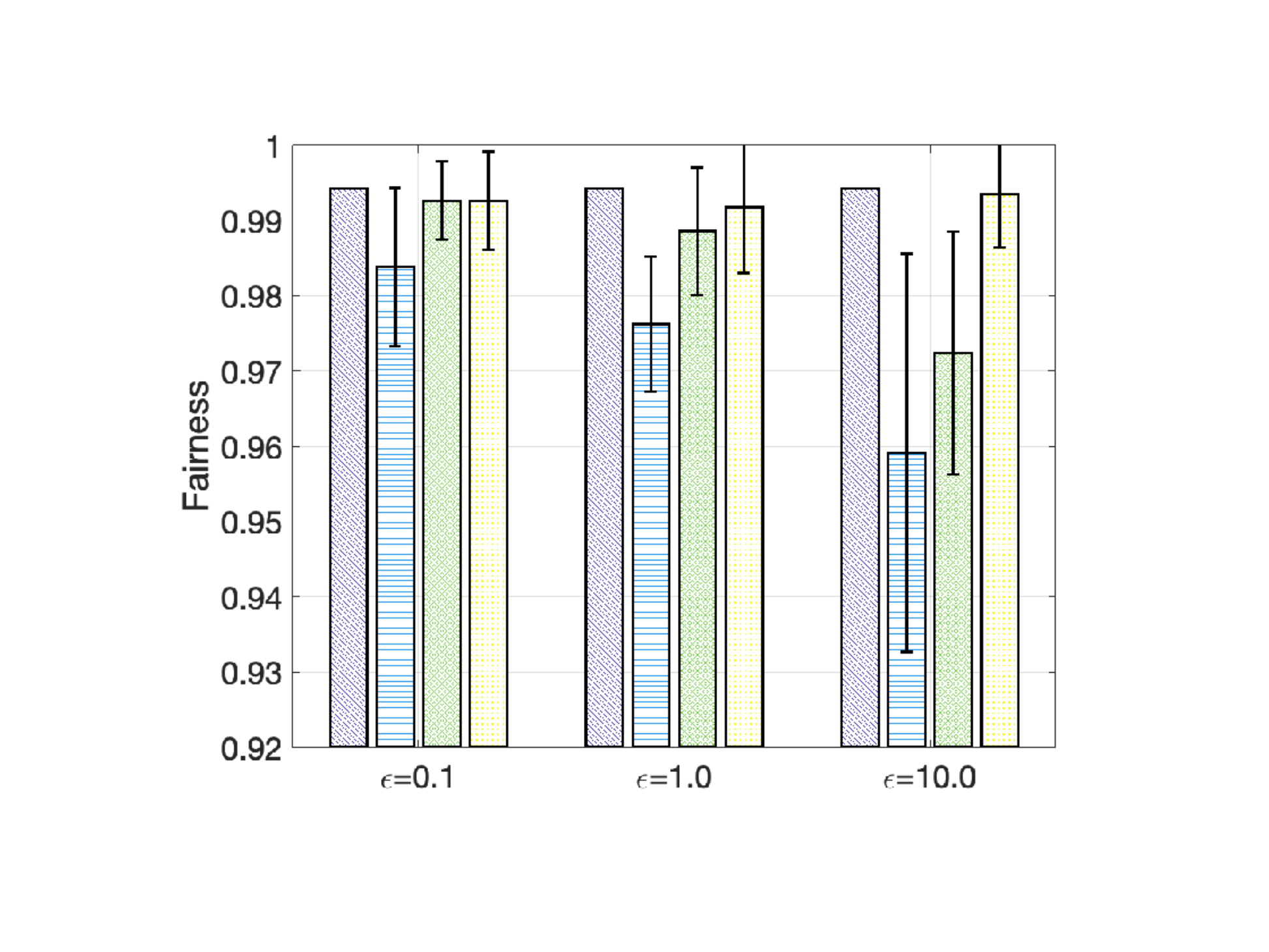}
    \includegraphics[width=0.49\linewidth,trim={2.5cm 2.cm 2.5cm 2.cm},clip]{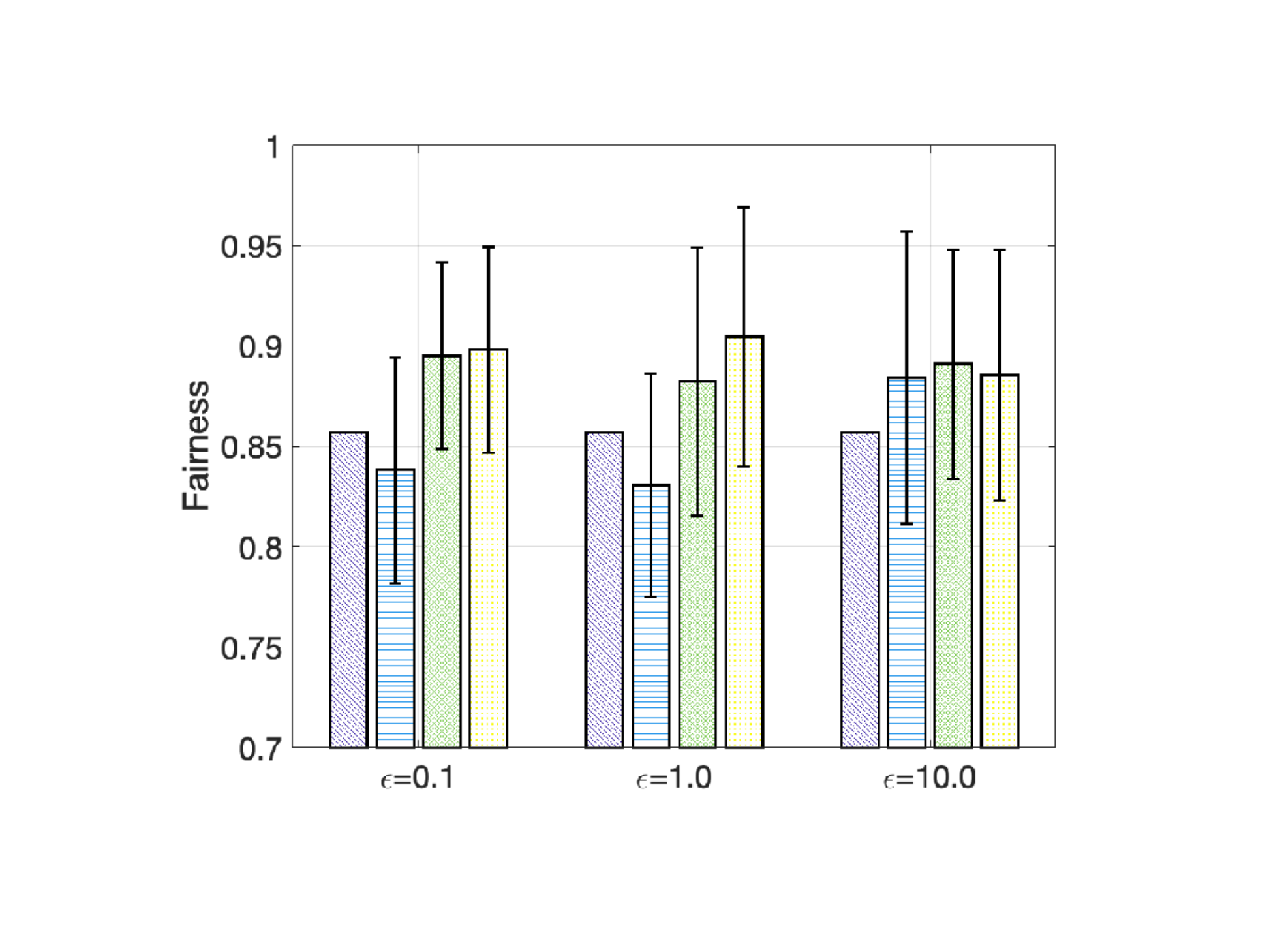}
    \caption{\nw{Fairness comparison between four scenarios and three privacy levels with threshold classifier ($t=0.8$) (left) and logistic linear regression classifier (right) on NCVR-mod dataset with age group as sensitive feature}}
    \label{fig:age_fairness}
    \end{minipage}\hfill
    \begin{minipage}{0.49\linewidth}
    \includegraphics[width=0.49\linewidth,trim={2.5cm 2.cm 2.5cm 2.cm},clip]{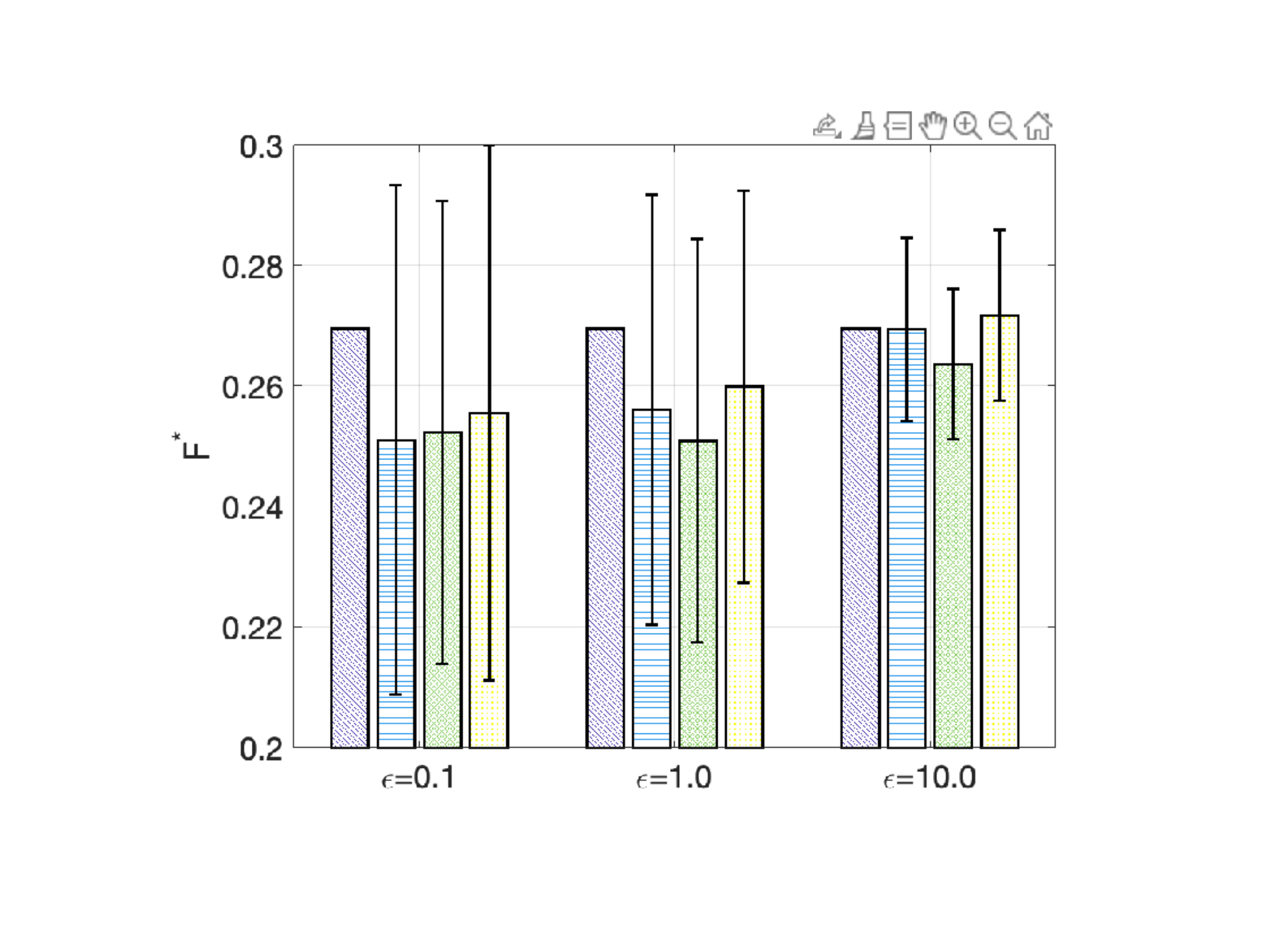}
    \includegraphics[width=0.49\linewidth,trim={2.5cm 2.cm 2.5cm 2.cm},clip]{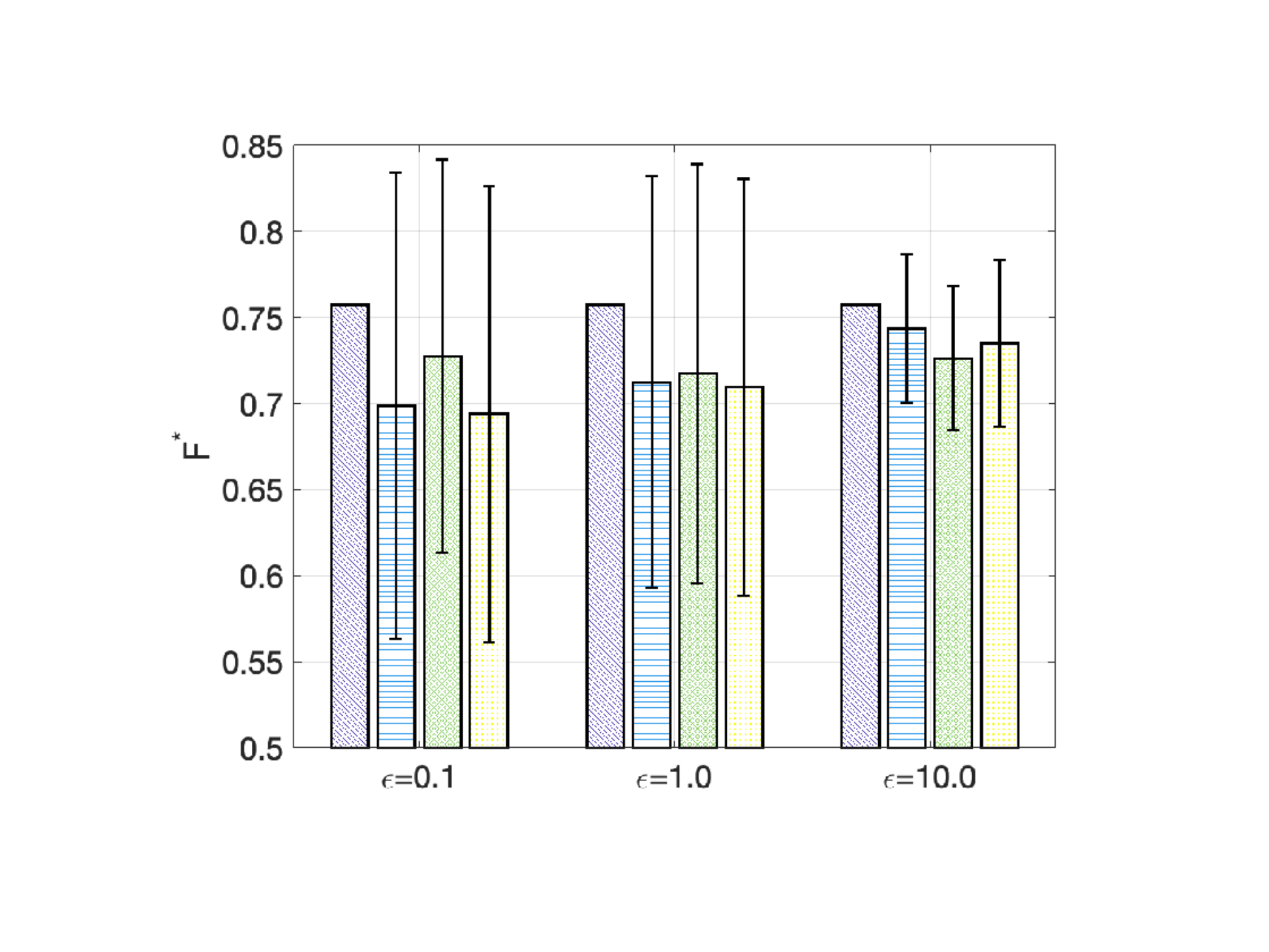}
   \caption{\nw{$F^*$-measure comparison between four scenarios and three privacy levels with threshold classifier ($t=0.8$) (left) and logistic linear regression classifier (right) on NCVR-mod dataset with age group as sensitive feature}}
    \label{fig:age_f1}
    \end{minipage}
\end{figure*}

\begin{figure*}[!t]
    \centering
    \RawFloats
    \begin{minipage}{0.49\linewidth}
    \includegraphics[width=0.49\linewidth,trim={2.5cm 2.cm 2.5cm 2.cm},clip]{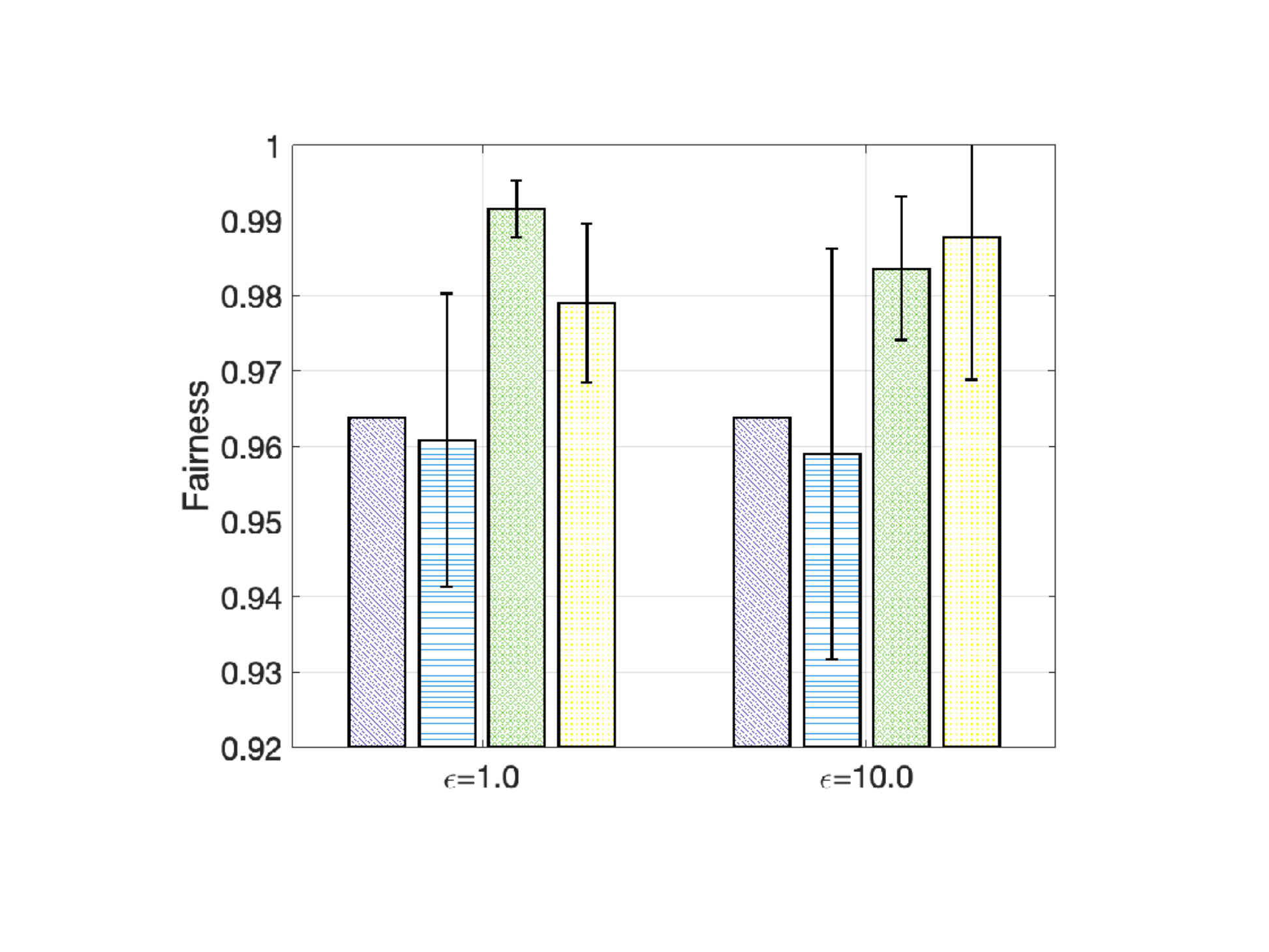}
    \includegraphics[width=0.49\linewidth,trim={2.5cm 2.cm 2.5cm 2.cm},clip]{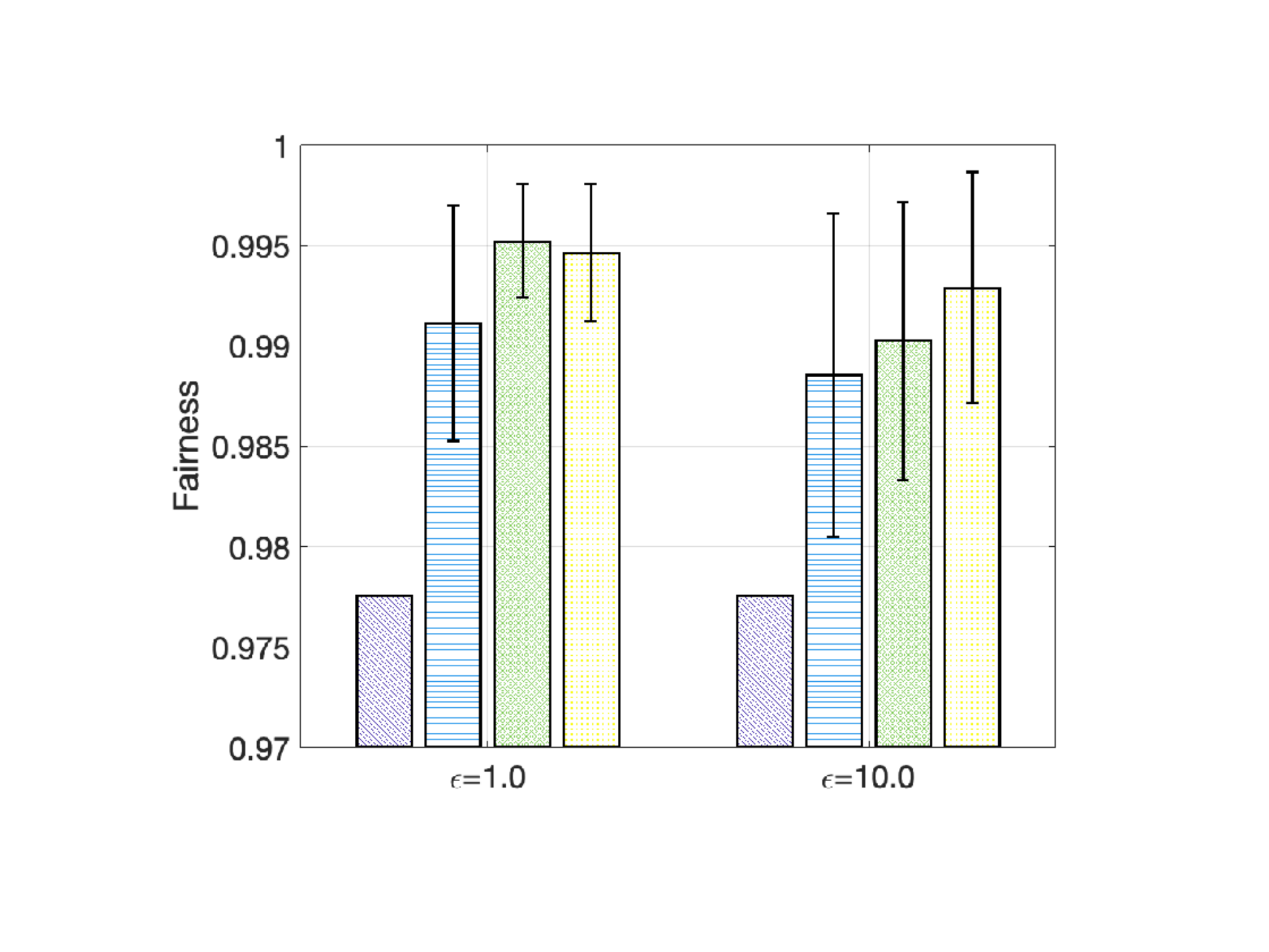}
   \caption{\nw{Fairness comparison between four scenarios and two privacy levels with threshold classifier ($t=0.8$) (left) and logistic linear regression classifier (right) on ABS dataset with gender as sensitive feature}}
    \label{fig:abs_fairness}
    \end{minipage}\hfill
    \begin{minipage}{0.49\linewidth}
    \includegraphics[width=0.49\linewidth,trim={2.5cm 2.cm 2.5cm 2cm},clip]{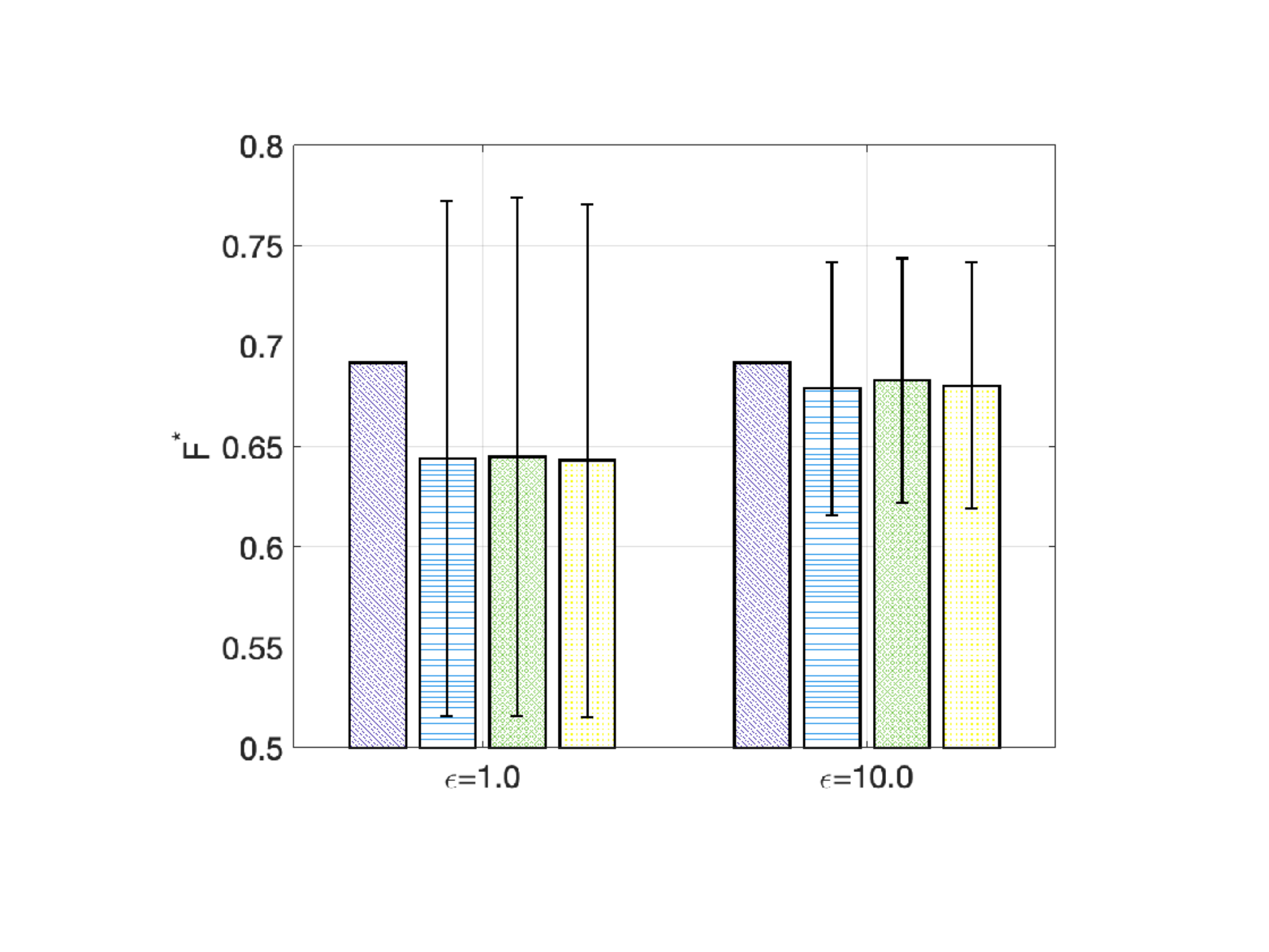}
    \includegraphics[width=0.49\linewidth,trim={2.5cm 2.cm 2.5cm 2cm},clip]{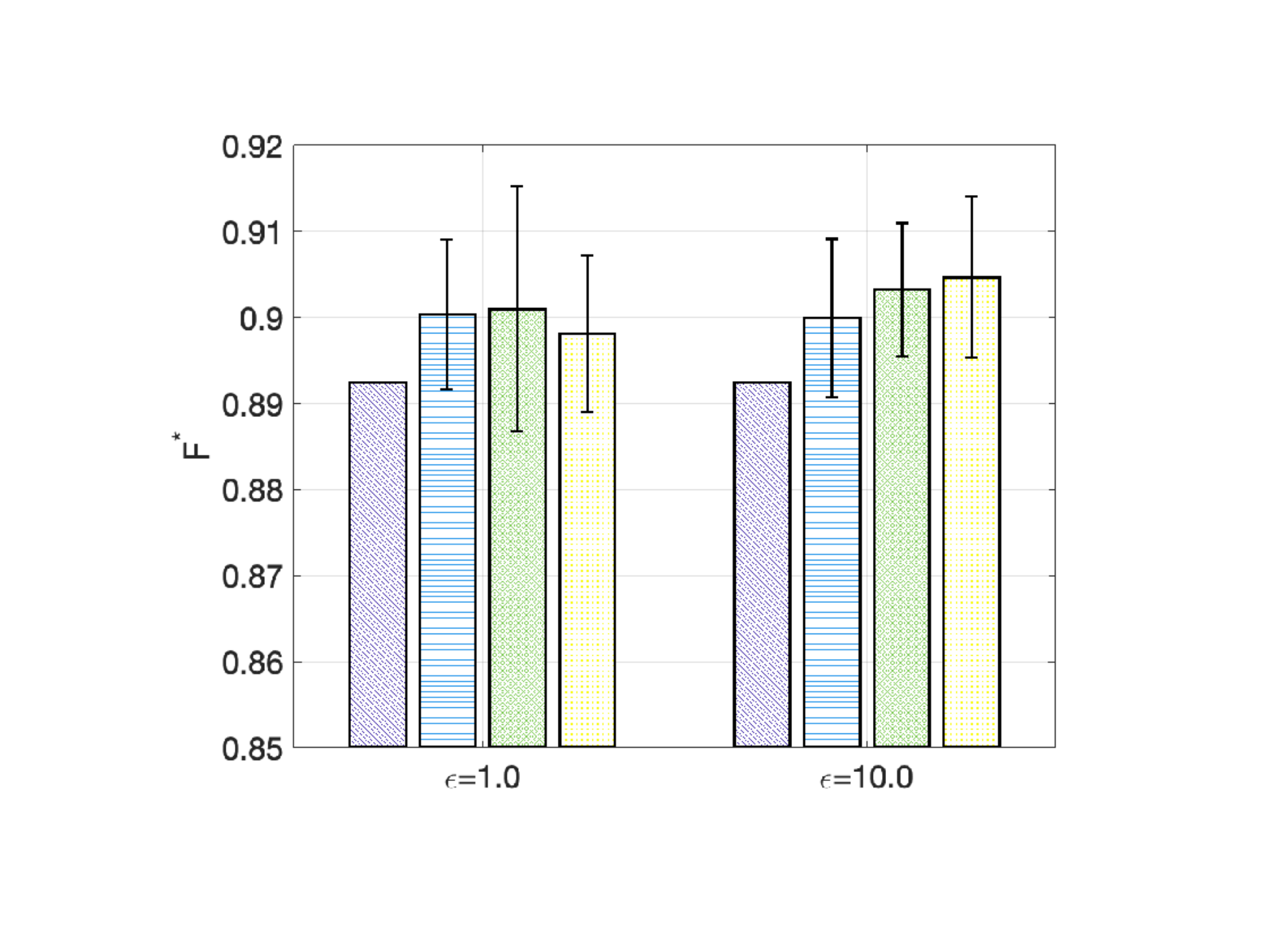}
   \caption{\nw{$F^*$-measure comparison between four scenarios and two privacy levels with threshold classifier ($t=0.8$) (left) and logistic linear regression classifier (right) on ABS dataset with gender as sensitive feature}}
    \label{fig:abs_f1}
    \end{minipage}
\end{figure*}

The false positive probability for one record pair with at least one dummy record with threshold classifier is shown in Fig.~\ref{fig:fp_probability}. In this figure, the theoretical prediction in blue curve matches to the results of real false positive probability in red plot. When flipping probability is equal to zero, which means the dummy record is a copy of its ancestor, the probability of this dummy record to be classified as a false match is $1.0$. With the increase of flipping probability, the false positive probability decreases. It is noted that when the flipping probability is greater than a certain value (e.g. around $0.28$ in Fig.~\ref{fig:fp_probability}), the False Positive probability reduces to zero. The behavior of False positive probability versus flipping probability matches to Theorem.~\ref{theorem:fp}.

Then, we evaluate the effect of privacy budget $\epsilon$ in false positive rate performance with threshold classifier $(t=0.8)$, while flipping probability is fixed for now. Recall the expression of False positive rate is $FPR=\frac{FP}{FP+TN}$. When Flipping probability is fixed, the false positive probability for one dummy record pair is fixed. So the numbers of FPs and TNs depend on the original dataset and the number of additional dummy records. As in Fig.~\ref{fig:fp_probability}, the value of false positive probability converges when flipping probability is great than $0.28$. So, to reduce the effect of flipping probability and its potential bias in record linkage, the values of the flipping probability for all gender groups ($G=2$ groups in our experiments) are ranging from $[0.4,1.0]$.

As shown Fig.~\ref{fig:fp_rate}, FPR is reversely proportional to the number of dummy records. It is remarkable that with the inclusion of of dummy records, FPR reduces as a reason of an increase in the number of TNs from dummy record pairs. When privacy budget $\epsilon$ increases, the FPR increases. Dummy record pairs can only be classified as either FP or TN, and there are always more TNs than FPs. This is because one dummy record is paired with many records other than the true match record from the other dataset. In Fig.~\ref{fig:fp_rate}, the blue curve is the theoretical results in Equation.~\ref{eq:fpr}. The red plot is the average empirical results with feature-level DP blocking method. The theoretical result matches with our empirical results of the feature-level DP method, and both of them validate the relationship between $\epsilon$ and FPR as in Equation.~\ref{eq:fpr}.
 
 

With the knowledge of the effect of flipping probability on false positive probability ($FPR$), we evaluate four scenarios:
 \begin{itemize}
     \item Baseline 1: no noise,
     \item Baseline 2: Feature-level DP blocking method,
     \item Method A: Fairness constrained feature-level DP blocking method,
     \item Method B: Cost constrained fairness-aware feature-level DP blocking method.
 \end{itemize}
In baseline 1, there are no noise added.  In baseline 2, the DP noise that are added for different protected features are the same. In other words, the privacy budget $\epsilon_g$ and flipping probability $flip_g$ for all protected feature $g\in[1,\cdots,G]$ are the same. In method A, privacy budget for each protected feature is kept as a constant value, while $flip_g\in[0.0,1.0]$ for group $g$ satisfies Definition.~\ref{defi:fairness-constrained}. In method B, flipping probabilities for each protected group are the same, while the privacy budget $\epsilon_g$ for group $g$ satisfies Definition.~\ref{defi:cost-constrained}.

With \nw{differential} private blocking method introduced to the record linkage process, there is an improvement in Fairness in terms of Equalised Odds from small privacy budget to large privacy budget as shown in Fig.~\ref{fig:nomod_fairness}, Fig.~\ref{fig:withmod_fairness}, Fig.~\ref{fig:age_fairness} and Fig.~\ref{fig:abs_fairness}. 
\dv{When gender is considered as sensitive feature, Fairness improves significantly for both threshold classifier ($t=0.8$) and logistic linear regression classifier on NCVR-Non-mod dataset, NCVR-mod dataset and ABS dataset}. \nw{In Fig.~\ref{fig:age_fairness}, when age group is used as a sensitive feature, fairness with differential private blocking improves only for logistic linear regression. This is because for threshold classifier, the fairness for four scenarios (including the baselines with no fairness) with small privacy budgets are already high and the differential privacy noise is stochastic with large privacy budget.
The fairness of our methods with respect to age group (in contrast to gender) does not always improve compared to baseline 1 without noise. However, our methods improved the fairness with respect to both age group and gender compared to baseline 2. Our results indicate that our methods with any protected feature (age group or gender) outperform the baselines by achieving both high privacy and high fairness.

}

Both Method A and Method B help reduce the fairness loss compared to Baseline 2 all the time. From Fig.~\ref{fig:nomod_fairness}, Fig.~\ref{fig:withmod_fairness}, \nw{Fig.~\ref{fig:age_fairness}} and Fig.~\ref{fig:abs_fairness}, Method B has better performance than Method A in reducing fairness-bias with respect to gender for large privacy budget, while Method A performs better with small privacy budget. The reason is when privacy budget is large the bias is more sensitive in cost. While privacy budget is small, the flipping probability value dominates the gender fairness-bias. Hence, Method A is more preferred than Method B for use when privacy budget is small.  
Baseline 1 with no noise added to the blocked bins experiences the lowest fairness performance all the time. 
\nw{Method B has better performance than Method A in reducing fairness-bias introduced by added differential private noise with age group as sensitive group. As shown in Fig.~\ref{fig:age_fairness}, fairness in Baseline 1 is relatively high with threshold classifier. The introduced DP blocking method increases fairness loss as shown in Baseline 2. Both Method A and Method B help in reducing the fairness loss in this case. It is remarkable that Method B has equivalent or better performance among four scenarios for both threshold classifier and logistic linear regression classifier.  }

\begin{figure*}[t]
    \centering
    \includegraphics[width=0.329\textwidth,trim={7cm 2.5cm 7cm 2cm},clip]{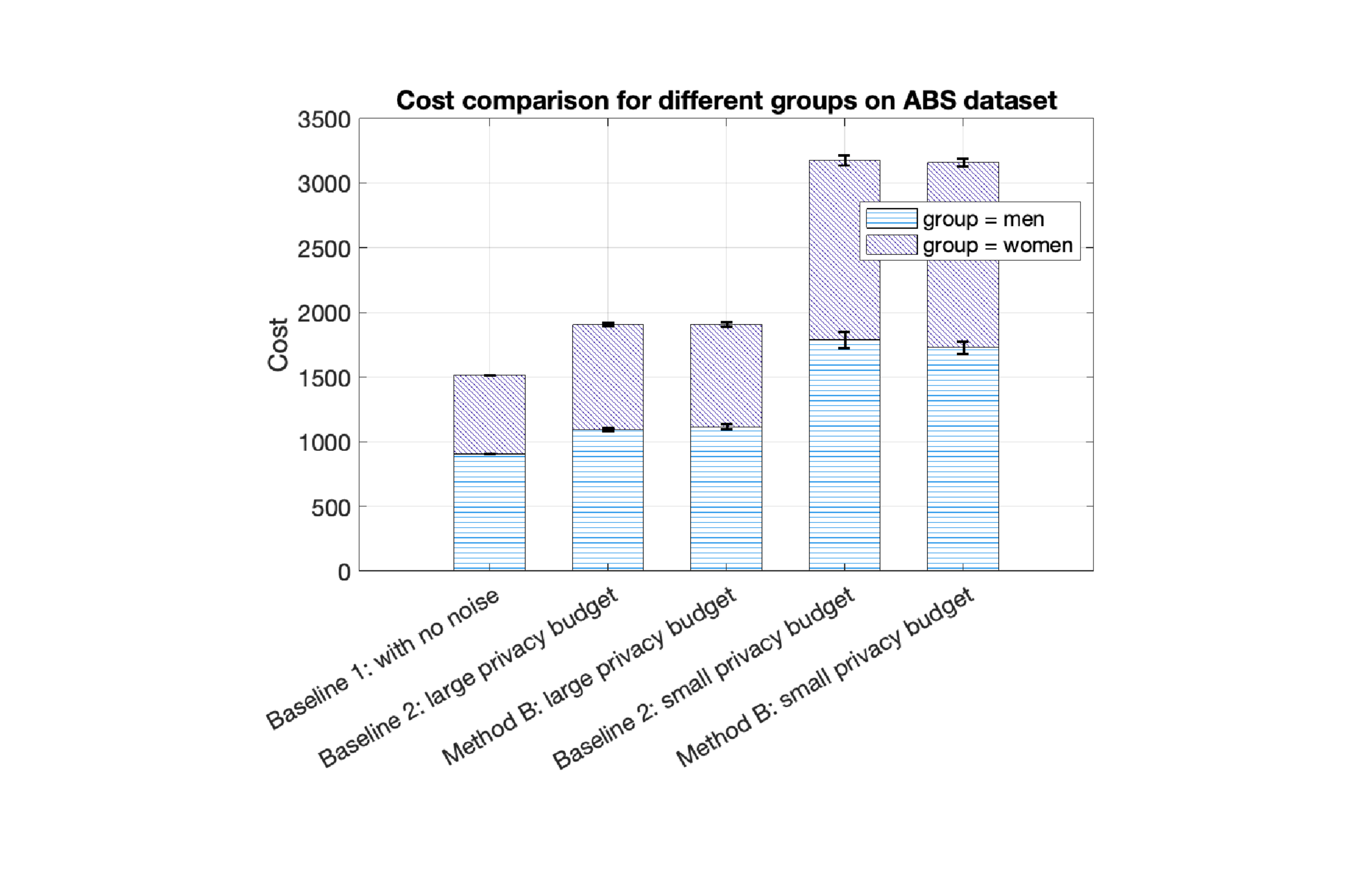}
    \includegraphics[width=0.329\textwidth,trim={7cm 2.5cm 7cm 2.cm},clip]{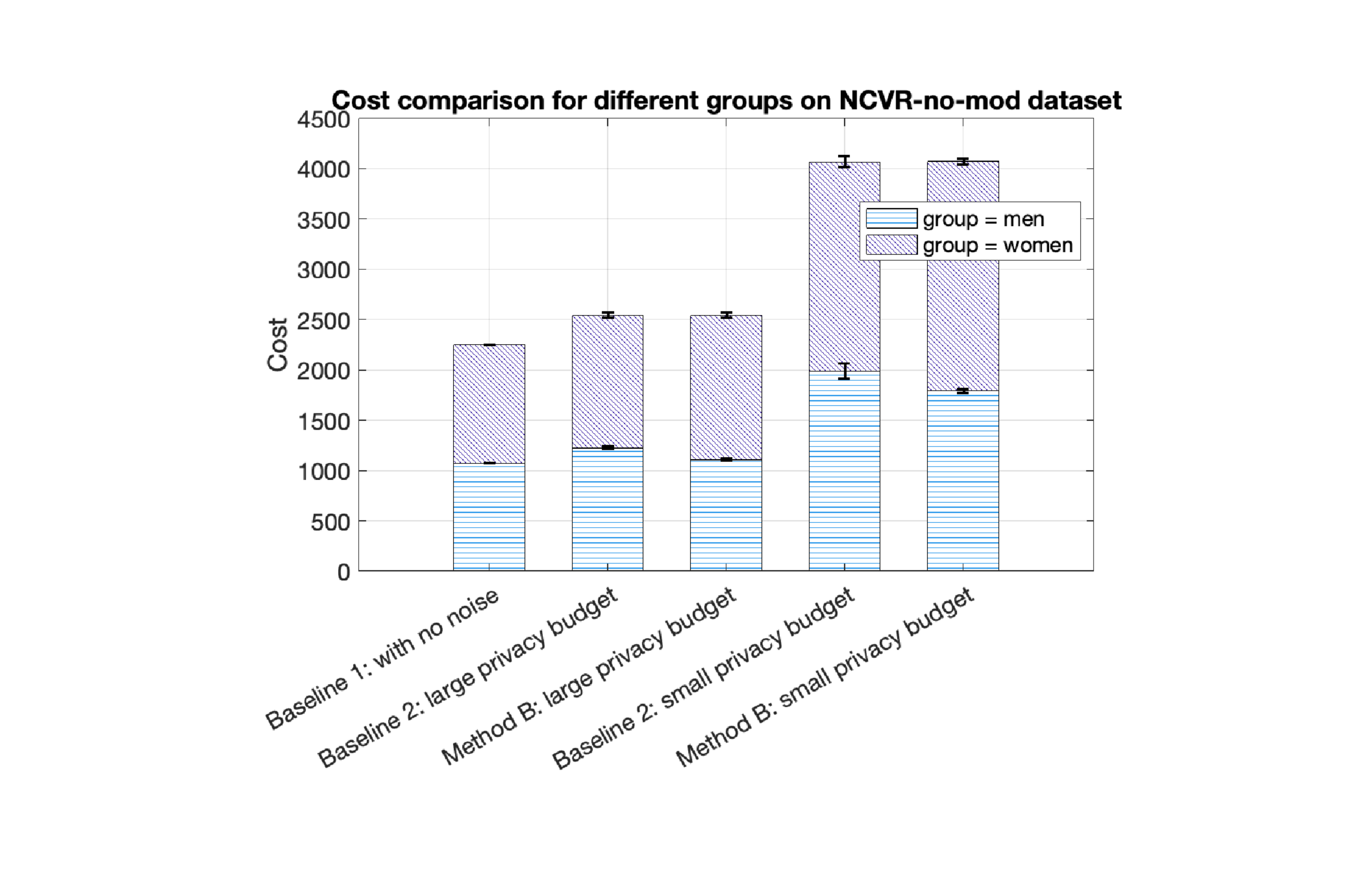}
    \includegraphics[width=0.329\textwidth,trim={7cm 2.5cm 7cm 2cm},clip]{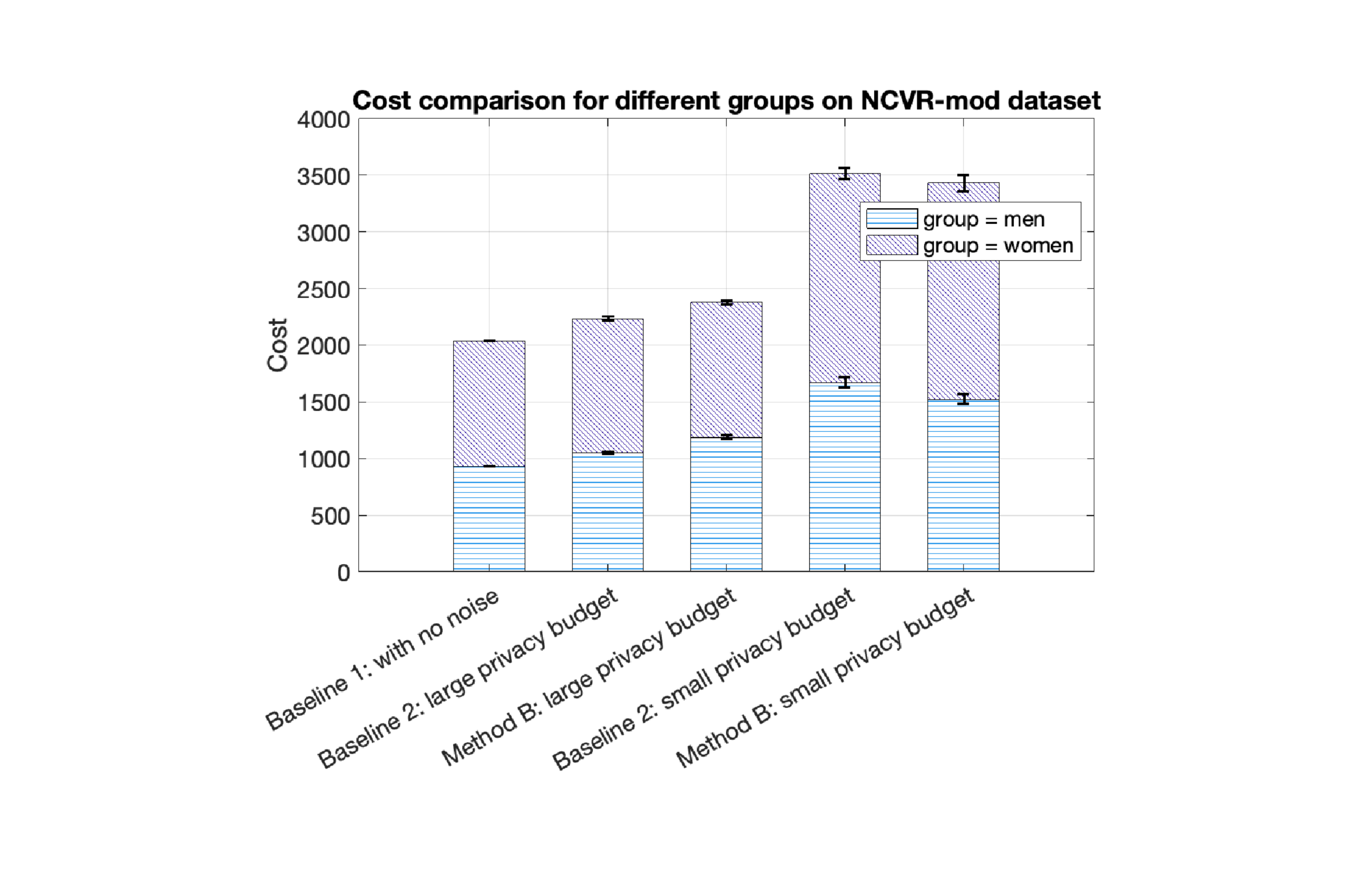}
    \caption{\nw{Cost comparison between Baseline 1: no noise, Baseline 2: feature-level DP blocking method and Method B Cost-constrained feature level DP blocking method for small and large privacy budgets on ABS dataset (left), NCVR-Non-mod dataset (middle) and NCVR-mod dataset (right)}}
    \label{fig:abs_cost}
\end{figure*}

Intuitively, adding privacy preserving noise to the bins of records seems to reduce the precision of the linkage performance. And this can be shown in Fig.~\ref{fig:nomod_f1}, Fig.~\ref{fig:age_f1} and Fig.~\ref{fig:abs_f1}. In Fig.~\ref{fig:nomod_f1}, when privacy budget is small $\epsilon=0.1$ and $\epsilon=1.0$ for threshold classifier, \nw{$F^*$-measure} for Baseline 2, Method A and Method B are worse than Baseline 1 no noise scenario.
\nw{When privacy budget $\epsilon$ is small, adding privacy preserving noise might reduce the $F^*$-measure as shown in Fig.~\ref{fig:nomod_f1} and in Fig.~\ref{fig:age_f1}.} 
\nw{While in Fig.~\ref{fig:abs_f1}, with $\epsilon=1.0$ and $\epsilon=10.0$ for threshold classifier, $F^*$-measure for Baseline 1 performs best in all four scenarios.} However, adding privacy preserving noise doesn't always decrease the \nw{$F^*$-measure} performance. Our results show that the \nw{$F^*$-measure} of scenarios with feature-level DP blocking methods are better than Baseline 1 with no noise \nw{in Fig.~\ref{fig:nomod_f1} on NCVR-Non-mod dataset with logistic linear regression classifier, and in Fig.~\ref{fig:withmod_f1} on NCVR-mod dataset with threshold classifier and logistic linear regression classifier, \nw{and in Fig.~\ref{fig:age_f1} with threshold classifier with large privacy budget,} and also in Fig.~\ref{fig:abs_f1} on ABS dataset with logistic linear regression classifier.  
}

It is shown that with Fairness-constrained feature-level DP blocking method and small privacy budget $\epsilon=0.1$ and $\epsilon=1.0$ with threshold classifier, both fairness and \nw{$F^*$-measure} are significantly improved compared to Baseline 1 and Baseline 2 in some \nw{high fairness-bias} cases. For logistic regression classifier on \nw{NCVR-no mod dataset}, NCVR-mod dataset and ABS dataset, Baseline 1 experiences lowest fairness and worst \nw{$F^*$-measure} performance among all scenarios \nw{in high fairness-bias cases}. \nw{Our methods perform better in terms of F* measure than both baselines with highly biased gender group. }

Method A and Method B improve Baseline 1 and Baseline 2 on fairness while requiring the same overall privacy budget.
For Method A, the privacy budgets between different protected features are the same, so the privacy budgets for all protected features remain the same as Baseline 2. By applying the ($flip_g$, Fairness)- Constrained Differential Privacy, the distance between False positive rates for male and female is reduced by adjusting the value of flipping probabilities, while the pairing cost for each protected feature remains the same as Baseline 2.  
For Method B, the overall privacy budget remains the same, while the privacy budgets for different protected features are different from each other. 
In method B, ($\epsilon$,Cost)-Constrained fairness-aware Differential Privacy is used. As shown in Fig.~\ref{fig:abs_cost}, the overall privacy budgets remain the same as Baseline 2, while the cost for protected group men/male reduces and the cost for protected group women/female increases.

\section{Conclusion}
\label{sec:conclusion}

Differential private grouping or blocking has been used in several works in the Privacy-Preserving Record Linkage (PPRL) literature to efficiently link (encoded) records from different parties while providing resilience against frequency inference attacks on the bins/blocks of encoded records. However, these methods use the standard differential privacy (DP) notion that does not consider other constraints, such as fairness of linkage and computational cost of comparing record pairs for linkage, when adding differential privacy noise to the bins. 

In this work, we propose new DP notions that are constrained not only on privacy guarantees, but also on fairness-bias in data and computational cost of linkage and apply our new DP notions to PPRL framework based on Bloom filter encoding and DP. We theoretically validate our new notions. Our experimental results show that the new PPRL algorithm following the two new DP notions constrained on fairness and cost provide better results in terms of fairness and cost for the same privacy guarantees. 

While our initial results are promising, the cost-constrained fairness-aware DP method does not perform well compared to the fairness-constrained DP method in terms fairness results. We would like to analyze further the cost-constraint method and the trade-off between cost and fairness.
In the future, we would like to explore fairness and cost-constrained DP for other learning tasks, including privacy-preserving active learning and privacy-preserving clustering. \dv{Another line of future work is experimenting the proposed notions with multiple protected features (e.g. gender and race) and different machine learning linkage models.}

\section*{Acknowledgment}
This research was funded by Macquarie University CyberSecurity Hub and strategic research funds from Macquarie University. Author Dinusha Vatsalan was affiliated with CSIRO Data61 at initial stages of the writing of this manuscript.


\ifCLASSOPTIONcaptionsoff
  \newpage
\fi



%




 \bibliographystyle{plain} 
\bibliography{IEEEabrv,paper}

\begin{IEEEbiography}[{\includegraphics[width=1in,height=1.25in,clip,keepaspectratio]{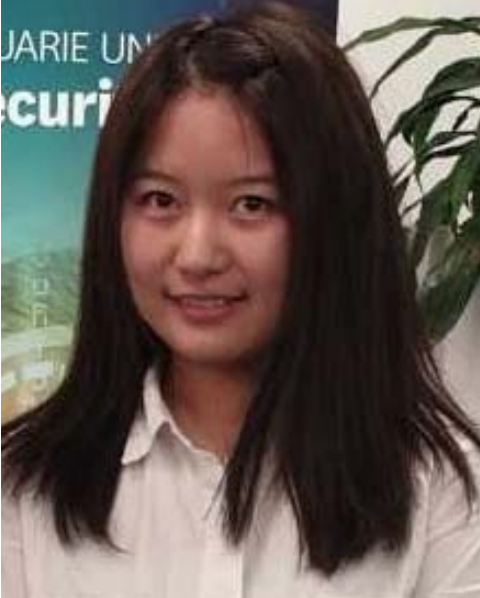}}]{Nan Wu} is currently pursuing the Ph.D. degree with Macquarie University and CSIRO's Data61, Australia. Nan received the B.S. degrees (Hons.) in electronic and communication systems from the Australian National University, Australia, in 2015, and from the Beijing Institute of Technology, China, in 2016, and the M.Res. degree in computer science from Macquarie University, Australia, in 2019. Her research interests include privacy-preserving machine learning, game theory, security and privacy, data sharing, data mining, and record linkage. 
\end{IEEEbiography}
\begin{IEEEbiography}[{\includegraphics[width=1in,height=1.25in,clip,keepaspectratio]{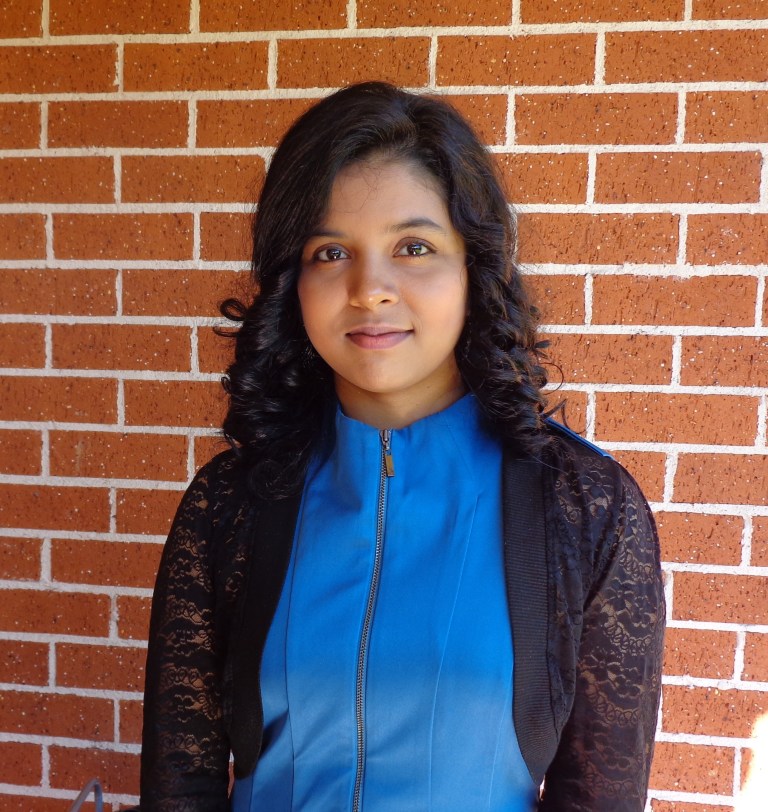}}]
{Dinusha Vatsalan} 
is a Senior Lecturer in Cyber Security at Macquarie University. Dinusha received her PhD in Computer Science from Australian National
University and BSc (Hons) from University of Colombo, Sri Lanka. She was
a Research Scientist at Data61, CSIRO.
Her research interests include privacy-preserving technologies for record linkage, data mining, machine learning, data sharing, and data
analytics, privacy attacks and defences, and privacy risk quantification. Dinusha has authored over 60 scientific articles in these research topics.
\end{IEEEbiography}


\begin{IEEEbiography}[{\includegraphics[width=1in,height=1.25in,clip,keepaspectratio]{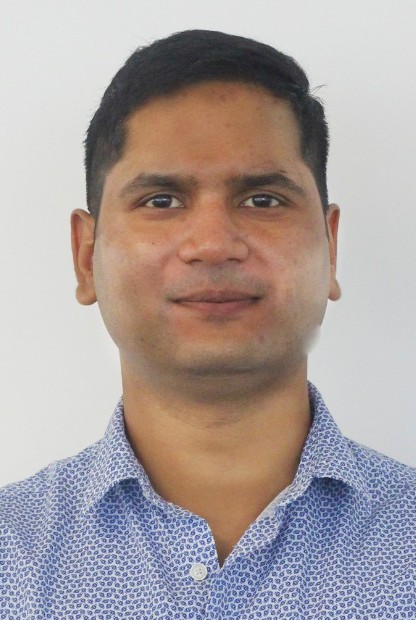}}]{Sunny Verma}
received his Ph.D. degree in Computer Science from the University of Technology Sydney in 2020. He is currently working as Postdoctoral Research Fellow on identifying cyberattacks on Human-bot teams within  AUSMURI at Macquarie University. He previously worked at the Data Science Institute, University of Technology Sydney, and Data61, CSIRO. His research interests include data mining, FATE, and interpretable deep learning.
\end{IEEEbiography}

\begin{IEEEbiography}[{\includegraphics[width=1in,height=1.25in,clip,keepaspectratio]{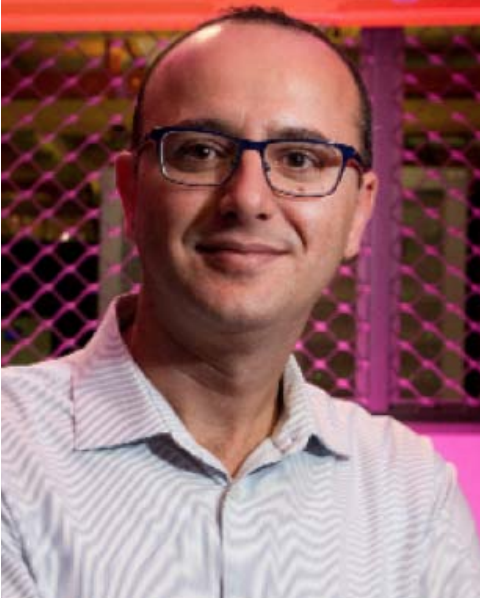}}]{Mohamed Ali (Dali) Kaafar}is a Professor at the Faculty of Science and Engineering at Macquarie University and the Executive Director of the Optus-Macquarie University Cyber Security Hub. He is also the founder of the Information Security and Privacy group at CSIRO Data61. Prior to that, Dali was the research leader of the Data Privacy and Mobile systems groups at NICTA and Senior principal researcher at INRIA, the French research institution of computer science and automation. He received his PhD from University of Nice Sophia Antipolis and INRIA in France where he pioneered research in the security of Internet Coordinate Systems. 

Prof. Kaafar is an associate editor of IEEE Transactions on Information Forensics \& Security and serves in the Editorial Board of the Journal on Privacy Enhancing Technologies. He published over 300 scientific peer-reviewed papers with several and repetitive publications in the prestigious IEEE Symposium on Security and Privacy (IEEE S\&P), ACM SIGCOMM, WWW, NDSS and PETS. 
He received several awards including the INRIA Excellence of research National Award, and the Andreas Pfitzman award from the Privacy Enhancing Technologies symposium in 2011. In 2019, he has also been awarded the prestigious and selective Chinese Academy of Sciences President's Professorial fellowship Award. 
\end{IEEEbiography}


\vfill



\end{document}